\documentclass{article}
\usepackage{lmodern}

\usepackage[margin=1.2in]{geometry}
\usepackage[utf8]{inputenc} 
\usepackage[T1]{fontenc}    

\newif\ifarxiv
\arxivtrue

\usepackage{hyperref}       
\usepackage{url}            
\usepackage{booktabs}       
\usepackage{enumerate}
\usepackage{enumitem}
\usepackage{amsmath,amsfonts,amsthm,amssymb}
\usepackage{leftindex}
\usepackage{centernot}
\usepackage{diagbox}

\usepackage{nicefrac}       
\usepackage{microtype}      
\usepackage{mathabx}
\usepackage{mathtools}
\usepackage{xfrac}
\usepackage{csquotes}
\usepackage{algorithm}
\usepackage{algpseudocode}

\usepackage{bbm}

\usepackage[]{xcolor}
\usepackage[numbers]{natbib}

\newcommand{\Ex}{\mathbb{E}}

\newcommand{\LR}{\Lambda}
\newcommand{\boost}{\mathrm{boost}}

\newtheorem{theorem}{Theorem}

\newtheorem{proposition}[theorem]{Proposition}

\newtheorem{remark}{Remark}
\newtheorem{example}{Example}

\usepackage{graphicx} 
\usepackage{caption}
\usepackage{subcaption}
\usepackage{enumitem}
\usepackage{cleveref}
\usepackage{numprint}
\usepackage{multirow}
\usepackage{makecell}

\graphicspath{{fig/}}

\title{Improving Wald's (approximate) sequential probability ratio test by avoiding overshoot}

\author{%
  Lasse Fischer and Aaditya Ramdas\\
   University of Bremen and Carnegie Mellon University \\
  \texttt{fischer1@uni-bremen.de,  aramdas@cmu.edu}
}

\begin{document}

\maketitle
\begin{abstract}
Wald's sequential probability ratio test (SPRT) is a cornerstone of sequential analysis. Based on desired type-I, II error levels $\alpha, \beta$, it stops when the likelihood ratio crosses certain  thresholds, guaranteeing optimality of the expected sample size. However, these thresholds are not closed form and the test is often applied with approximate thresholds $(1-\beta)/\alpha$ and $\beta/(1-\alpha)$ (approximate SPRT). When $\beta > 0$, this neither guarantees error control at $\alpha,\beta$ nor optimality. 
When $\beta=0$ (power-one SPRT), this method is conservative and not optimal. The looseness in both cases is caused by \emph{overshoot}: the test statistic overshoots the thresholds at the stopping time. Numerically calculating thresholds may be infeasible, and most software packages do not do this.
We improve the approximate SPRT by modifying the test statistic to avoid overshoot. Our `sequential boosting' technique \emph{uniformly} improves power-one SPRTs $(\beta=0)$ for simple nulls and alternatives, or for one-sided nulls and alternatives in exponential families. When $\beta > 0$, our techniques
provide guaranteed error control at $\alpha,\beta$, while needing less samples than the approximate SPRT in our simulations. We also provide several nontrivial extensions: confidence sequences, sampling without replacement and conformal martingales. 

\end{abstract}


\section{Introduction}

To begin with a simple setting that we relax significantly later on, suppose we observe a stream of i.i.d. data $X_1,X_2,\ldots$ and are interested in a sequential test for
\begin{align}
H_0: X_i\sim \mathbb{P}_0 \quad \text{vs.} \quad H_1: X_i\sim \mathbb{P}_1. \label{eq:simple_test}
\end{align}
In particular, suppose we seek to control the type I and II error at level $\alpha,\beta$ for some predefined $\alpha,\beta \in [0,1]$.
The classical approach to this testing problem is the sequential probability ratio test (SPRT) introduced by \citet{wald1945sequential, Wald1947}. Define the likelihood ratio 
 \begin{align}
 \LR_0:=1 \quad \text{and} \quad \LR_t:=\frac{\prod_{i=1}^t p_1(X_i)}{\prod_{i=1}^t p_0(X_i)} \ (t\in \mathbb{N}), \label{eq:likelihood_ratio}
 \end{align}
 where $p_0$ and $p_1$ are the densities with respect to some measure $\mu$ of $ \mathbb{P}_0$ and $ \mathbb{P}_1$, respectively. For some chosen constants $0\leq \gamma_0 < \gamma_1$, at each time $t$, Wald's SPRT does the following: 
\begin{align*}
    &\text{ stop and reject $H_0$ if }  \LR_t\geq \gamma_1,\\
    &\text{ stop and accept $H_0$ if } 
    \LR_t\leq \gamma_0, \\
    &\text{ continue sampling if } 
    \gamma_0<\LR_t<\gamma_1.
\end{align*}
 
Let $\delta_{\gamma_0,\gamma_1} \in \{0,1\}$ be the test decision of the SPRT (1 if $H_0$ is rejected and 0 if accepted) and $\tau_{\gamma_0,\gamma_1}$ the stopping time. \citet{wald1944cumulative} showed that $\tau_{\gamma_0,\gamma_1}$ is almost surely finite. Furthermore, \citet{wald1948optimum} proved that for any other sequential test $(\delta', \tau')$ it holds that \begin{align} \mathbb{P}_h(\delta'=1-h)\leq \mathbb{P}_h(\delta_{\gamma_0,\gamma_1}=1-h) \ (h\in  \{0,1\}) \implies \mathbb{E}_h(\tau')\geq \mathbb{E}_h(\tau_{\gamma_0,\gamma_1}) \ (h\in  \{0,1\}).\label{eq:wald_wolfo}\end{align} 
Hence, the SPRT minimizes the expected sample size among all other sequential tests with the same or smaller type I and type II error. While this may seem to be the end of the story, the catch is that it is difficult to specify $\gamma_0,\gamma_1$ so that the error probabilities match desired levels. We expand on this below.

\paragraph{Wald's approximate SPRT}
In practice, it may be quite difficult to calculate $\mathbb{P}_0(\delta_{\gamma_0,\gamma_1}=1)$ and $\mathbb{P}_1(\delta_{\gamma_0,\gamma_1}=0)$ exactly \citep{wald1945sequential, siegmund2013sequential}. Therefore, we cannot choose $\gamma_0$ and $\gamma_1$ such that the type I error is equal to a desired level $\alpha\in (0,1)$ and the type II error is equal to a desired level $\beta \in (0,1)$.  For this reason, \citet{wald1945sequential, Wald1947} introduced approximations for $\gamma_1$ and $\gamma_0$:
\begin{align}\gamma_1=\frac{1-\beta}{\alpha} \text{ and }\gamma_0=\frac{\beta}{1-\alpha}.\label{eq:Walds_approx}\end{align}
\citet{wald1945sequential} showed that at least one of the two errors will be controlled at the desired level:
\begin{align}
\mathbb{P}_0(\delta_{\gamma_0,\gamma_1}=1)+\mathbb{P}_1(\delta_{\gamma_0,\gamma_1}=0)\leq \alpha + \beta. \label{eq:inequality_Wald}
\end{align}
 Nevertheless, in many trials it is required that at least the type I error is controlled at the prespecified level $\alpha$. Furthermore, even if the sum of type I and type II errors are controlled, it is not ensured that the SPRT with Wald's approximation minimize the expected sample size, meaning more efficient procedures are possible. The reason for the inaccuracy of Wald's approximation is that it is derived under the assumption of no \emph{overshoot}, meaning that at the stopping time $\tau_{\gamma_0,\gamma_1}$ the likelihood ratio equals exactly one of the boundaries $\gamma_0$ or $\gamma_1$. 
 
 To summarize, \textit{the thresholds in~\eqref{eq:Walds_approx} neither guarantee type I, II error control nor optimal use of sample size, and yet they are used very frequently in practice} (e.g. \citep{spiegelhalter2003risk, ivan2004application, chetouani2014sequential, hapfelmeier2023efficient}). In particular, the \texttt{R} packages \enquote{sprtt: Sequential Probability Ratio Tests Toolbox} \citep{steinhilber2023} and \enquote{SPRT: Wald's Sequential Probability Ratio Test} \citep{Bottine2015} apply these approximate boundaries (the R Package \enquote{SPRT} was archived and removed from CRAN in 2022). 

\paragraph{The conservative power-one SPRT} 
If we set $\beta=0$ in Wald's approximated thresholds \eqref{eq:Walds_approx}, then the SPRT stops for rejection if $\LR_t\geq 1/\alpha$, but never stops for accepting $H_0$. Interestingly, the SPRT provably controls the type I error at level $\alpha$ in this case. This follows, for example, by Ville's inequality which we will recap in Section~\ref{sec:avoid_overshoot}. In addition, it is still guaranteed that the SPRT stops almost surely after a finite number of samples if $H_1$ is true \citep{wald1944cumulative}, and therefore is called conservative power-one SPRT, in particular studied and espoused by Robbins~\cite{darling1968some}.
A level-$\alpha$ sequential test of power one (in short, a \emph{power-one sequential test}) is defined by a stopping time $\tau$ (at which you reject the null) such that
\begin{align}
\mathbb{P}_0(\tau<\infty)\leq \alpha \label{eq:power-one_0} \text{ and } \mathbb{P}_1(\tau<\infty)=1,
\end{align}
where $\alpha\in (0,1)$ is a prespecified significance level.
The power-one SPRT $\tau_{\gamma_1}:=\inf\{t\in \mathbb{N}: \LR_t\geq \gamma_1\}$, $\gamma_1\geq 1$, uses the available samples optimally in the following sense \citep{chow1971siegmund}:
\begin{align*}
\mathbb{P}_0(\tau'<\infty) \leq \mathbb{P}_0(\tau_{\gamma_1}<\infty) \implies \mathbb{E}_1(\tau')\geq \mathbb{E}_1(\tau_{\gamma_1}).
\end{align*}
Consequently, any other sequential test that controls the type I error at the same level the SPRT does, needs on average at least as many samples as the SPRT to make a rejection if the alternative is true. 
Since it may be difficult to calculate $\mathbb{P}_0(\tau_{\gamma_1}<\infty)$ exactly, one often employs the conservative power-one SPRT by setting $\gamma_1=1/\alpha$ to ensure \eqref{eq:power-one_0}. But this only guarantees exact type I error control and thus optimality if $(\LR_t)_{t\in \mathbb{N}}$ never overshoots at $\gamma_1$ and almost surely attains $0$ or $\gamma_1$ at some point \citep{ramdas2020admissible}. However, this is usually not fulfilled.

\paragraph{Wald's conservative SPRT}
The aforementioned Ville's inequality also implies that the type II error for Wald's two-sided SPRT is bounded by $\gamma_0$. Therefore, the SPRT with \[\gamma_1=1/\alpha \text{ and } \gamma_0=\beta\] 
 provides valid type I and type II error control at $\alpha$ and $\beta$ respectively, which is why we refer to it as Wald's conservative SPRT in the following. 

\paragraph{Outline of the paper and our contributions}

In this paper, we propose a general \emph{sequential boosting method} to uniformly improve (super)martingale based sequential tests by avoiding an overshoot at $1/\alpha$ (Section~\ref{sec:avoid_overshoot}). This means that our boosted sequential tests are never worse than the initial tests, but may stop earlier and are more likely to achieve a rejection.
In Section~\ref{sec:SPRT}, we use this to uniformly improve conservative power-one SPRTs for simple and one-sided nulls and alternatives. We provide nontrivial extensions and applications of our method in Section~\ref{sec:extensions}, improving some existing confidence sequences, methods for sampling without replacement and conformal martingales. In Section~\ref{sec:futility}, we demonstrate how our boosting approach should be applied if a stop for futility is to be included and type II error control is desired. Finally, we discuss extensions of our method and existing approaches for dealing with overshoot when using the SPRT (Section~\ref{sec:discussion}). All formal proofs are provided in the supplementary material.

\begin{table}[h!]
    \centering
    \begin{tabular}{c c c}
      Sequential test & Error control & Sample size  \\ \hline
        Conservative power-one SPRT 
        & $\alpha'\leq\alpha, \beta'=0$ & Not optimal  \\[0.2cm]
       Our boosted power-one SPRT & $\alpha'\leq \alpha, \beta'= 0$ & \makecell{Not optimal, but never stops \\ later than power-one SPRT}  \\
       \hline
        Wald's approximate SPRT 
        & \makecell{$\alpha'\leq\frac{\alpha}{1-\beta}, \beta'\leq\frac{\beta}{1-\alpha}$ \\ $\alpha'+ \beta'\leq \alpha + \beta$} & Not optimal  \\[0.6cm]
        Wald's conservative SPRT 
        & $\alpha'\leq \alpha, \beta'\leq \beta$ & Not optimal  \\[0.2cm]
       Our boosted two-sided SPRT & $\alpha'\leq \alpha, \beta'\leq \beta$ & \makecell{Not optimal, but often stops \\ earlier than approx.\ SPRT and \\ never later than cons.\ SPRT}  \\
        \hline
    \end{tabular}
    \caption{Comparison of the guarantees provided by Wald's approximate and conservative SPRTs, to our boosted SPRTs. The desired type I and type II error probabilities are denoted by $\alpha$ and $\beta$, while the actually achieved probabilities are denoted by $\alpha'$ and $\beta'$. 
    }
    \label{tab:summary}
\end{table}

In Table~\ref{tab:summary} we compare Wald's SPRTs with our boosted versions. The type I error is denoted by $\alpha'$ and the type II error by $\beta'$. In the power-one case, our SPRT provides the same guarantees as Wald's SPRT, but uniformly improves it in terms of sample size. For the two-sided case, Wald's conservative SPRT and our boosted SPRT provide valid type I and type II error control while Wald's approximate SPRT does not. In addition, our two-sided SPRT outperformed Wald's approximate SPRT with respect to sample size in all considered simulation scenarios and provably improves Wald's conservative SPRT. 

\paragraph{Comparison to existing alternatives of Wald's approximated thresholds}
We do \emph{not} claim that our boosted tests are optimal in the sense of \eqref{eq:wald_wolfo}. Hence, if the exact parameters $\gamma_0$ and $\gamma_1$ for obtaining a desired type I and type II error guarantee with the SPRT are known, one should employ the usual SPRT with these thresholds. For example, this can be the case in settings where the (asymptotic) approximations by \citet{siegmund2013sequential} are reliable (see Section~\ref{sec:discussion} for a more detailed discussion of that approach). However, we will see that our approach can be applied in a much broader range of settings than Siegmund's approximations, including in some nonparametric settings (like testing exchangeability using conformal martingales), and provides non-asymptotic control of the type I and type II error. Further advantages over Siegmund's and other approaches, such as the use of simulations to determine the exact thresholds, are discussed in Section~\ref{sec:other_approachs}.

\section{Improving test martingales by avoiding overshoot\label{sec:avoid_overshoot}}

A process $(M_t)_{t\in \mathbb{N}_0}$ adapted to some filtration $(\mathcal{F}_t)_{t\in \mathbb{N}_0}$ (in the following we denote processes and filtrations just by $(M_t)$ and $(\mathcal{F})$, respectively) is a test 
martingale \citep{vovk2005algorithmic} for $\mathbb{P}_0$ if $(M_t)$ is nonnegative, $\Ex_0[M_0]=1$ and $\Ex_0[M_t|\mathcal{F}_{t-1}]=M_{t-1}$.  We will also speak of test supermartingale, if the last equality is replaced by an inequality. Ville's inequality \citep{ville1939etude} states that for any test supermartingale $(M_t)$ it holds that
 \begin{align}
 \mathbb{P}_0(\exists t\in \mathbb{N}: M_t \geq 1/\alpha)\leq \alpha\quad (\alpha\in (0,1)).\label{eq:ville}
 \end{align}
 For example, the likelihood ratio process $(\Lambda_t)$ \eqref{eq:likelihood_ratio} is a test martingale, since 
 \begin{align}\mathbb{E}_{0}\left[\Lambda_t|\Lambda_{t-1}\right]=\Lambda_{t-1}\mathbb{E}_{0}\left[\frac{p_{1}(X_t)}{p_{0}(X_t)}\right] = \Lambda_{t-1}\int_{\mathcal{X}} \frac{p_{1}(x)}{p_{0}(x)} p_{0}(x) \mu(dx)= \Lambda_{t-1},\label{eq:LR_test_mart}\end{align}
where $\mathcal{X}$ is the support of $X_t$ under $\mathbb{P}_0$ (one could set $\mathcal{F}_{t}=\sigma(X_1,\ldots,X_t)$, but $\Lambda_t$ depends on $\mathcal{F}_{t-1}$ only through $\Lambda_{t-1}$ which is why we simply condition on $\Lambda_{t-1}$). This proves the type I error control of the conservative power-one SPRT ($\gamma_1=1/\alpha$,  $\gamma_0=0$). However, Ville's inequality holds with equality only if $(M_t)$ never overshoots at $1/\alpha$ \citep{ramdas2020admissible} and therefore this power-one SPRT usually does not exhaust the type I error. We now introduce a general approach to uniformly improve (super)martingale based sequential tests by avoiding overshoot.

\subsection{Truncation and boosting of test supermartingales}

Suppose we have some nonnegative process $(M_t)$ with $\mathbb{E}_0[M_0]=1$ that is adapted to a filtration $(\mathcal{F}_t)$. For example, one can think of $\mathcal{F}_t=\sigma(X_1,\ldots,X_t)$ as being generated by the observations, however, in some cases it is useful to coarsen the filtration or add randomization, as we shall later see with conformal martingales. It should be noted that in this section we do not make any assumption about the distribution of the data (in particular, $X_1,X_2,\ldots$ do not need to be i.i.d.). We define $L_t=M_t/M_{t-1}$ with the convention $0/0=0$ as the individual multiplicative factors of $(M_t)$, so that
\[
M_t = M_{t-1}\cdot L_t = \prod_{i=1}^t L_i\quad (t\in \mathbb{N}).
\]
Then $(M_t)$ is a test supermartingale with respect to $(\mathcal{F}_t)$ iff \begin{align}\mathbb{E}_{0}[L_t|\mathcal{F}_{t-1}]\leq 1\quad  \text{for all } t\in \mathbb{N}.\label{eq:growth_factor}\end{align}
Hence, every test supermartingale can be decomposed into its individual factors. Furthermore, if factors $(L_t)_{t\in \mathbb{N}}$ adapted to $(\mathcal{F}_t)$ satisfying \eqref{eq:growth_factor} are given, we can construct a new test supermartingale by $M_t=\prod_{i=1}^t L_i$. 
Defining the stopping time $$\tau_M:=\inf\{t\in \mathbb{N}_0:M_t\geq 1/\alpha\}$$ at which we reject the null, Ville's inequality implies type I error control: $\mathbb{P}_0(\tau_M < \infty) \leq \alpha$. Mathematically, a stopping time $\tau$ is an integer-valued random variable such that $\{\tau\leq t\}$ is measurable with respect to $\mathcal{F}_t$ for all $t\in \mathbb{N}$.  
 The optional stopping theorem implies that $\mathbb{E}_0[M_{\tau}]\leq 1$ for all stopping times $\tau$. 


Our approach is to increase the individual factors $L_t$ of a given test supermartingale $(M_t)$, yet avoiding an overshoot at $1/\alpha$. As a first step, we define the truncation function
\begin{align}
    T_{\alpha}(x; M)=\begin{cases}
        x, &\text{ if } Mx\leq \frac{1}{\alpha } \\
        \frac{1}{M \alpha}, &\text{ if } Mx> \frac{1}{\alpha }.
    \end{cases}
\label{eq:truncation}
\end{align}
Note the key property that $M \cdot T_{\alpha}(x; M) \leq 1/\alpha$ for all $x\in [0,\infty]$. 
Now define the process 
\begin{align}
M_t^{\mathrm{trunc}}= M_{t-1}^{\mathrm{trunc}} \cdot T_{\alpha}(L_i; M_{t-1}^{\mathrm{trunc}}) = \prod_{i=1}^t T_{\alpha}(L_i;M_{i-1}^{\mathrm{trunc}}).
\end{align}
By construction, $M_t^{\mathrm{trunc}}$ never overshoots at $1/\alpha$. 
Now, define
\[
\tau_{\mathrm{trunc}}:=\inf\{t\in \mathbb{N}_0:M_t^{\mathrm{trunc}} \geq 1/\alpha\}.
\]
We claim that $$\tau_{\mathrm{trunc}}=\tau_M,$$ meaning the sequential test based on $(M_t^{\mathrm{trunc}})$ makes the same decision and at the same time as the test based on $(M_t)$.
This is because for $t < \tau_M$, $M_t^{\mathrm{trunc}} = M_t$, meaning that the truncation has no effect if the threshold is not reached. At the very last step, when $M_t$ reaches (but may potentially exceed) $1/\alpha$, $M_t^{\mathrm{trunc}}$ simply equals $1/\alpha$, proving our claim.

However, since $T_{\alpha}(L_t;M^{\mathrm{trunc}}_{t-1})\leq L_t$, we can potentially increase the factors $T_{\alpha}(L_t;M^{\mathrm{trunc}}_{t-1})$ and thus stop sooner while retaining valid error control. We propose to improve $L_t$ by multiplying it with a \enquote{boosting factor} $b_t\geq 1$ (the larger, the better): at each step $t$, we calculate $b_t$ as large as possible such that 
\begin{align}\Ex_{0}[T_{\alpha}(b_t L_t;M_{t-1}^{\mathrm{boost}})|\mathcal{F}_{t-1}]\leq 1,\label{eq:boosting_inequality}\end{align}
where $M_{t-1}^{\mathrm{boost}}=\prod_{i=1}^{t-1} T_{\alpha}(b_iL_i; M_{i-1}^{\mathrm{boost}})$. Now, $(M_{t}^{\mathrm{boost}})$ is our boosted process that we use for testing. This procedure is summarized in Algorithm~\ref{alg:general}. Note that if $M_{t-1}^{\mathrm{boost}}=1/\alpha$, we could choose $b_t=\infty$, so we just set $M_{t}^{\mathrm{boost}}=1/\alpha$ in this case.

\begin{algorithm}
\caption{Improving test supermartingales by avoiding overshoot} \label{alg:general}
 \hspace*{\algorithmicindent} \textbf{Input:} Test supermartingale $(M_t)$ and individual significance level $\alpha$.\\
 \hspace*{\algorithmicindent} 
 \textbf{Output:} Boosted test supermartingale $(M_t^{\mathrm{boost}})$.
\begin{algorithmic}[1]
\State $M_0^{\mathrm{boost}}=1$
\For{$t=1,2,\ldots$}
  \State Define $L_t:=M_{t}/M_{t-1}$
  \State Choose $b_t\geq 1$ as large as possible such that $\Ex_{0}[T_{\alpha}(b_t L_t;M_{t-1}^{\mathrm{boost}})|\mathcal{F}_{t-1}]\leq 1$
  \State $M_t^{\mathrm{boost}}=M_{t-1}^{\mathrm{boost}} T_{\alpha}(b_t L_t;M_{t-1}^{\mathrm{boost}})$
\EndFor
\State \Return $(M_t^{\mathrm{boost}})$
\end{algorithmic}
\end{algorithm}

\begin{theorem}\label{theo:main}
    Let $(M_t)$ be any test supermartingale and $(M_t^{\mathrm{boost}})$ be the boosted process obtained by Algorithm~\ref{alg:general}, yielding the sequential tests $\tau_M=\inf\{t\in \mathbb{N}: M_t\geq 1/\alpha\}$ and $\tau_{\mathrm{boost}}=\inf\{t\in \mathbb{N}: M_t^{\mathrm{boost}}\geq 1/\alpha\}$, respectively. Then $M_t^{\mathrm{boost}}\geq M_t$  for all $t<\tau_M$ and $\tau_{\mathrm{boost}}\leq \tau_M$. Furthermore, $(M_t^{\boost})$ is a test supermartingale. Therefore, under the null,  $\mathbb{P}_0(\tau_{\mathrm{boost}}<\infty)\leq \alpha$ and $\mathbb{E}_0[M_{\tau}^{\boost}]\leq 1$ for all stopping times $\tau$. 
\end{theorem}

Theorem~\ref{theo:main} shows that Algorithm~\ref{alg:general} provides a universal improvement of test supermartingales if $\alpha$ is prespecified. Regardless of which test supermartingale $(M_t)$ is put into Algorithm~\ref{alg:general}, the test based on the output process $(M_t^{\boost})$ always stops sooner ($\tau_{\boost} \leq \tau_M$). Further, if we stop before the null is rejected, the boosted test martingale provides more evidence against the null  than the initial process ($M_t^{\boost}\geq M_t$ for all $t<\tau_M$).

\begin{remark}
 The individual factors $L_t$ have sometimes been termed `sequential e-values'  in the recent literature \cite{ramdas2023game, vovk2024merging}. A similar, but non-sequential, truncation approach like \eqref{eq:truncation} was proposed by \citet{wang2022false} in the context of multiple testing. They also mentioned that in general we can use any \enquote{boosting function} $\psi_t:[0,\infty]\to [0,\infty]$ such that
 \begin{align}\mathbb{E}_{0}[T_t(\psi_t(L_t);M_{t-1}^{\boost})|\mathcal{F}_{t-1}]\leq 1 \quad \text{for all } t\in \mathbb{N}.\end{align}
 However, in a different context, parallel work by \citet{koning2024continuous}  proved that $\psi_t(x)=b_tx$, where $b_t$ is chosen such that \eqref{eq:boosting_inequality} is satisfied with an equality, is the log-optimal choice. Therefore, we will stick to this particular boosting factor for the rest of the paper.
\end{remark}

\begin{remark}
   $(M_t^{\mathrm{boost}})$ is a test martingale if \eqref{eq:boosting_inequality} holds with equality.
    Therefore, $(M_t^{\mathrm{boost}})$ can be interpreted as a likelihood ratio between $\mathbb{Q}$  and $\mathbb{P}_0$, where $\mathbb{Q}$ is adaptively determined. 
\end{remark}

\subsection{Existence and calculation of optimal boosting factors}

One question is, of course, if there is a largest possible boosting factor $b_t$ and, if so, how we can determine it. In the following proposition, we show that there always exists a $b_t^*$ that satisfies \eqref{eq:boosting_inequality} with an equality and that if we have found such a $b_t^*$, then there is no further $b_t$ which leads to a larger boosted martingale factor. 

\begin{proposition}\label{prop:solvability}
    If $\mathbb{P}_0(L_t=0|\mathcal{F}_{t-1})=0$ and $M_{t-1}^{\boost}\in (0,\alpha^{-1})$, then there exists $b_t^*$ s.t. 
    \begin{align}\Ex_{0}[T_{\alpha}(b_t^* L_t;M_{t-1}^{\mathrm{boost}})|\mathcal{F}_{t-1}]= 1.\label{eq:boosting_equality}\end{align}
    For every $b_t\geq 1$ satisfying \eqref{eq:boosting_inequality}, we have $T_{\alpha}(b_t^* L_t;M_{t-1}^{\mathrm{boost}})\geq T_{\alpha}(b_t L_t;M_{t-1}^{\mathrm{boost}})$, $\mathbb{P}_0$-a.s.
\end{proposition}

 Since $T_{\alpha}(x ; M_{t-1}^{\boost})$ is nondecreasing in $x$, it is numerically very easy to find $b_t^*$;  start with $b_t= 1$ and slightly increase it until \eqref{eq:boosting_equality} is satisfied with the desired accuracy. This can be solved accurately with standard root solving algorithms (we use the \texttt{R} function \textit{uniroot} in our implementations). The only requirement is that we know the distribution of $L_t|\mathcal{F}_{t-1}$ under $\mathbb{P}_0$ such that we can calculate the expected value in \eqref{eq:boosting_equality} for arbitrary $b_t$. In the following section we provide a closed form for this expected value for SPRTs. In some cases (e.g., for binary data), it is even possible to obtain a closed form for $b_t^*$ (see Section~\ref{sec:WoR}).

\section{Improving power-one SPRTs\label{sec:SPRT}}

In this section, we demonstrate how (conservative) power-one SPRTs can be improved by the boosting technique proposed in Algorithm~\ref{alg:general}. Hence, we consider the setting from the introduction and assume that $X_1,X_2,\ldots$ are i.i.d. data. SPRTs are particularly well-suited for our method, since each individual factor $L_t$ has the simple form $L_t=\lambda(X_t)$, where $\lambda$ is the likelihood ratio between the null and the alternative. In addition, the distribution of $\lambda(X_t)$ under the null hypothesis is usually known.

Although only required in Subsection~\ref{sec:comp_null}, to unify the notation we assume that the distribution class under consideration can be parametrized by some parameter space $\Theta$. The true but unknown parameter of interest will be denoted by $\theta^*\in \Theta$.

\subsection{Simple null vs. simple alternative\label{sec:simple_null_alt}}
 In this subsection, we consider the simple testing problem
$$
H_0:\theta^*=\theta_0 \quad \text{vs.} \quad H_1:\theta^*=\theta_1,\ \theta_1\neq \theta_0.
$$

In this case, the likelihood ratio is given by $\lambda(X_t)=\frac{p_{\theta_1}(X_t)}{p_{\theta_0}(X_t)}$. This allows to write the expected value in \eqref{eq:boosting_inequality} in the following simple form 
\begin{align}
    &\mathbb{E}_{\theta_0}[T_{\alpha}(b_t \lambda(X_t) ; \Lambda_{t-1}^{\boost})| \Lambda_{t-1}^{\boost}] \nonumber \\ &=\mathbb{E}_{\theta_0}\left[b_t \lambda(X_t)\mathbbm{1}\left\{b_t \lambda(X_t)\leq \frac{1}{\alpha \Lambda_{t-1}^{\boost}}\right\} + \frac{1}{\Lambda_{t-1}^{\boost}\alpha}\mathbbm{1}\left\{b_t \lambda(X_t)> \frac{1}{\alpha \Lambda_{t-1}^{\boost}}\right\}\Bigg\vert \Lambda_{t-1}^{\boost}\right] \nonumber \\
    &= b_t\int_{\mathcal{X}} \frac{p_{\theta_1}(x)}{p_{\theta_0}(x)} p_{\theta_0}(x) \mathbbm{1}\left\{b_t \lambda(x)\leq \frac{1}{\alpha \Lambda_{t-1}^{\boost}}\right\} \mu(dx)  + \frac{1}{\Lambda_{t-1}^{\boost}\alpha}\mathbb{P}_{\theta_0}\left( \lambda(X_t)> \frac{1}{b_t \alpha \Lambda_{t-1}^{\boost}}\bigg\vert \Lambda_{t-1}^{\boost}\right)  \nonumber \\
    &=b_t\mathbb{P}_{\theta_1}\left(\lambda(X_t)\leq \frac{1}{b_t\alpha \Lambda_{t-1}^{\boost}}\bigg\vert \Lambda_{t-1}^{\boost}\right)  + \frac{1}{\Lambda_{t-1}^{\boost}\alpha}\mathbb{P}_{\theta_0}\left( \lambda(X_t)> \frac{1}{b_t \alpha \Lambda_{t-1}^{\boost}}\bigg\vert \Lambda_{t-1}^{\boost}\right),\label{eq:boost_SPRT}
\end{align}
where $\Lambda_{t-1}^{\boost}=\prod_{i=1}^{t-1} T_{\alpha}(b_i \lambda(X_i);\Lambda_{i-1}^{\boost})$. To calculate an optimal boosting factor, we can then just set the equation equal to $1$ and solve for $b_t$ numerically (see Proposition~\ref{prop:solvability}).

For example, in case of a simple Gaussian testing problem $H_0: X_i\sim \mathcal{N}(\mu_0,1)$ vs. $H_1: X_i\sim \mathcal{N}(\mu_1,1)$, $\mu_1-\mu_0=\delta>0$, we obtain
$$
\lambda(x)=\exp(\delta (x-\mu_0)-\delta^2/2) \quad \text{and} \quad \lambda^{-1}(y)=\frac{\delta^2/2+\log(y)}{\delta}+\mu_0.
$$
Hence, \eqref{eq:boost_SPRT} becomes
\begin{align}
b_t\left[\Phi\left(\frac{\log\left([b_t\alpha \Lambda_{t-1}^{\boost}]^{-1}\right)}{\delta} - \frac{\delta}{2}\right)\right] +(\Lambda_{t-1}^{\boost}\alpha)^{-1} \left[1-\Phi\left(\frac{\log([b_t\alpha \Lambda_{t-1}^{\boost}]^{-1})}{\delta}
+ \frac{\delta}{2}\right)\right], \label{eq:Gaussian_boosting}
\end{align}
where $\Phi$ is the CDF of a standard Gaussian. 

For instance, suppose we test the above hypothesis for $\delta=2$ at level $\alpha=0.05$ and we are at the beginning of the testing process, meaning $\Lambda_{t-1}^{\boost}=\Lambda_{0}^{\boost}=1$. Then we obtain a boosting factor of $b_t=1.28$. If the current martingale value increases, the boosting factor increases as well, such that we obtain $b_t=1.55$ for  $\Lambda_{t-1}^{\boost}=2$, but $b_t=1.14$   for $\Lambda_{t-1}^{\boost}=0.5$. In a similar manner, $b_t$ is increasing in $\delta$. 

See Table~\ref{tab:boosting_factors} for boosting factors on a grid of values for $\delta$ and $\Lambda_{t-1}^{\boost}$. If $\delta$ and $\Lambda_{t-1}^{\boost}$ are both small, the boosting factors are close to one. The reason is that in this case the probability under $H_0$ of overshooting at $1/\alpha$ in the next step is close to zero. But also note that for small $\delta$ the expected sample size is usually large and this table only shows the boosting factors for one step. The total gain of boosting at some time could be calculated as the cumulative product of the boosting factors up to that time. The table also shows that the SPRT is very conservative for large $\delta$. For example, if $\delta=3$ and $\lambda(X_1)=\lambda(X_2)=1$, then $\Lambda_2=1$, but we would already have $\Lambda_2^{\boost}>20$ such that the boosted test could already reject the hypothesis. 

\begin{table}[!htb]
    \centering
    \begin{tabular}{l|c c c c c }
        \diagbox[]{$\delta$}{$\Lambda_{t-1}^{\boost}$} & $0.5$ & $1$ & $2$ & $4$ & $10$ \\ \hline
        $0.1$ & $1$ & $1$ & $1$ & $1$ & $1.00001$ \\ 
        $0.5$ & $1$ & $1$ & $1$ & $1.00019$ & $1.03019$ \\
        $1$ & $1.00015$ & $1.00157$ & $1.01077$ & $1.05386$ & $1.37349$ \\ 
        $2$ & $1.13931$ & $1.27600$ & $1.55046$ & $2.17468$ & $5.73972$ \\ 
        $3$ & $2.45490$ & $3.49439$ & $5.72975$ & $11.8255$ & $68.1985$ \\ 
    \end{tabular}
    \caption{Boosting factors obtained for $\alpha=0.05$ in a simple Gaussian testing problem with mean difference $\delta$ and variance $1$. \label{tab:boosting_factors}}
\end{table}

In Figure~\ref{fig:sim_simple} we compare the SPRT with its boosted improvement regarding sample size and type I error in the simple Gaussian testing setup for different strengths of the signal $\delta$, where $\mu_0=0$. The results were obtained by averaging over $\numprint{10000}$ trials. In each trial, we generated up to $\numprint{10000}$ observations under the alternative until the hypothesis could be rejected by the respective sequential test.  Both tests rejected the hypothesis in all cases. 
In order to determine the type I error for both of the methods, we used importance sampling by exploiting the fact \citep{siegmund1976importance} that

\begin{align}
\mathbb{P}_{\theta_0}(\tau <\infty) = \mathbb{E}_{\theta_1}\left[\mathbbm{1}\{\tau <\infty\}   \Lambda_{\tau}^{-1}\right], \label{eq:importance_sampling}
\end{align}
where $\Lambda_{\tau}^{-1}$ is the likelihood ratio of $\theta_0$ to $\theta_1$ at time $\tau$.
With this, we can estimate the type I error by the inverse of the likelihood ratio at the stopping time while sampling under the alternative. Since $\tau$ is almost surely finite and usually much smaller under the alternative than under the null hypothesis, this reduces the required number of samples to obtain a good estimate of the type I error \citep{siegmund1976importance, siegmund2013sequential}.

The boosting approach saves on average $3.4 \%$ to $11.8 \%$ of the sample size compared to the SPRT in the considered scenarios, while the percentage increases with the strength of the signal. Furthermore, the boosted SPRT exhausts the significance level, while the SPRT is becoming increasingly conservative for stronger signals. For all considered $\delta$, the mean computational time required for running one boosted SPRT, including the calculation of the corresponding boosting factors, was below $0.05$ seconds (performed on a laptop with an Intel Core i7-10510U CPU, 1.80GHz, 8 cores, 16 GB RAM), showing that the approach is computationally feasible in practice.  
In Supplementary Material~\ref{appn:sims} we repeat this experiment for $\alpha=0.01$ instead of $\alpha=0.05$, which yields a similar result.

\begin{figure}[h!]
\centering
\includegraphics[width=\textwidth]{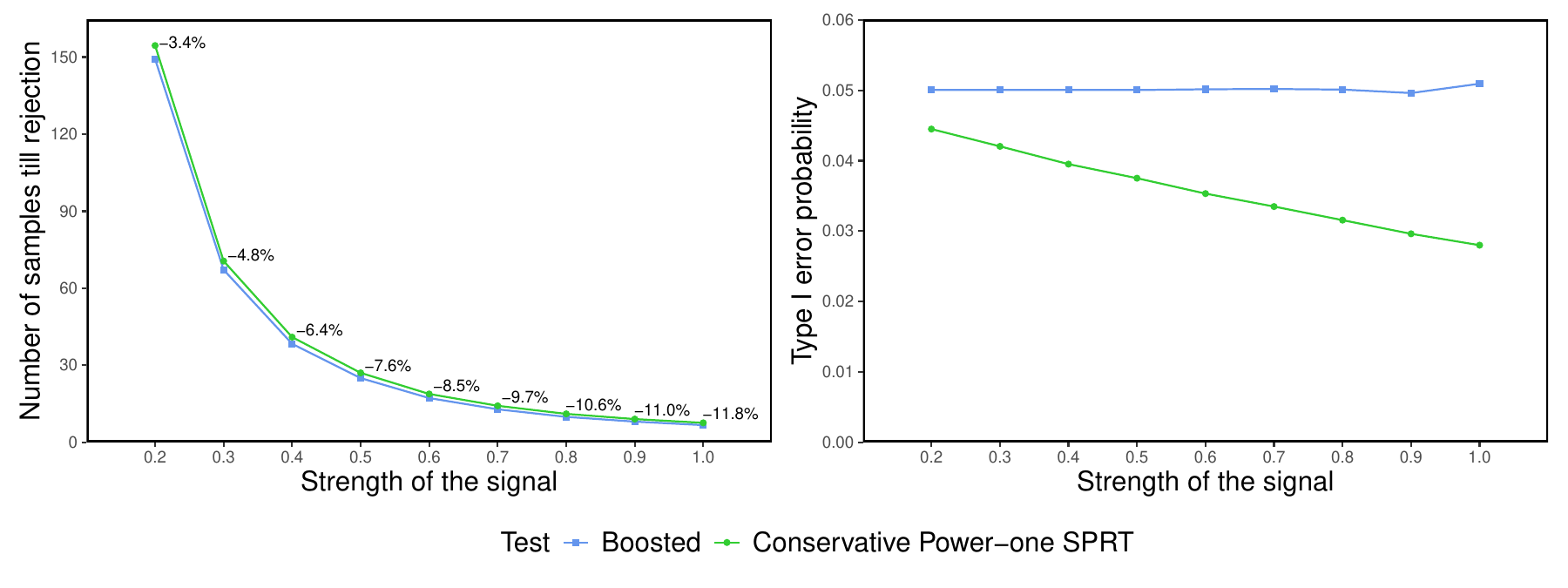}
\caption{Comparison of sample size and type I error between the SPRT and the boosted SPRT in a simple Gaussian testing setup for different signal strengths and $\alpha=0.05$. The boosted SPRT saves $3.4-11.8\%$ of sample size compared to the SPRT and exhausts the significance level. The SPRT is becoming increasingly conservative for stronger signals.  \label{fig:sim_simple} }\end{figure}

\subsection{Handling composite alternatives with predictable plugins\label{sec:comp_alt}}

In practice, a concrete alternative may be hard to specify, leading to problems of the form
$$
H_0: \theta^* = \theta_0 \quad \text{vs.} \quad H_1:\theta^* >\theta_0.
$$
There are two approaches that are usually used to handle such composite alternatives in sequential testing: the method of mixtures \citep{robbins1970statistical} and the predictable plugin \citep{wald1945sequential, robbins1973class, robbins1974expected}. In the following, we will recap the latter, since it is better suited for our boosting technique. 

Let $(\theta_t)_{t\in \mathbb{N}}$ be a predictable process, meaning $\theta_t$ is measurable with respect to $\mathcal{F}_{t-1}$,  such that $\theta_t\geq \theta_0$ for all $t\in \mathbb{N}$ and define
\begin{align}
\Lambda_t^{\text{plugin}}=\prod_{i=1}^t \frac{p_{\theta_i}(X_i)}{p_{\theta_0}(X_i)}. \label{eq:Z_t}
\end{align}
 The predictable plugin method replaces the numerator of the likelihood ratio the density of a fixed alternative by an estimated density $p_{\theta_t}$ that is based on the previous data $X_1,\ldots, X_{t-1}$.  By replacing $\Lambda_t$ with $\Lambda_t^{\text{plugin}}$ in \eqref{eq:LR_test_mart},
it follows immediately that $(\Lambda_t^{\text{plugin}})$ defines a test martingale for $H_0$. Hence, the application of Algorithm~\ref{alg:general} is straightforward: at each $t\in \mathbb{N}$, we calculate a boosting factor using \eqref{eq:boost_SPRT} by replacing $\theta_1$ with $\theta_t$.


A simple  predictable plugin is given by the smoothed maximum likelihood estimator
\begin{align}
\theta_t=\max\left(\frac{\theta_0 + \sum_{i=1}^{t-1} X_i}{t},\theta_0\right). \label{eq:plugin_robbins}
\end{align}
 In Figure~\ref{fig:sim_plugin} we repeated the experiment from Figure~\ref{fig:sim_simple} in Section~\ref{sec:simple_null_alt} but assumed that the alternative is not specified in advance using the predictable parameter in \eqref{eq:plugin_robbins}. Unsurprisingly, the sample size required by both methods is larger than when the alternative was prespecified. However, the gain in sample size due to boosting is even larger in the case of the composite alternative, saving $6.3\%$ to $13.9\%$ of the sample size compared to the usual SPRT.

\begin{figure}[h!]
\centering
\includegraphics[width=0.8\textwidth]{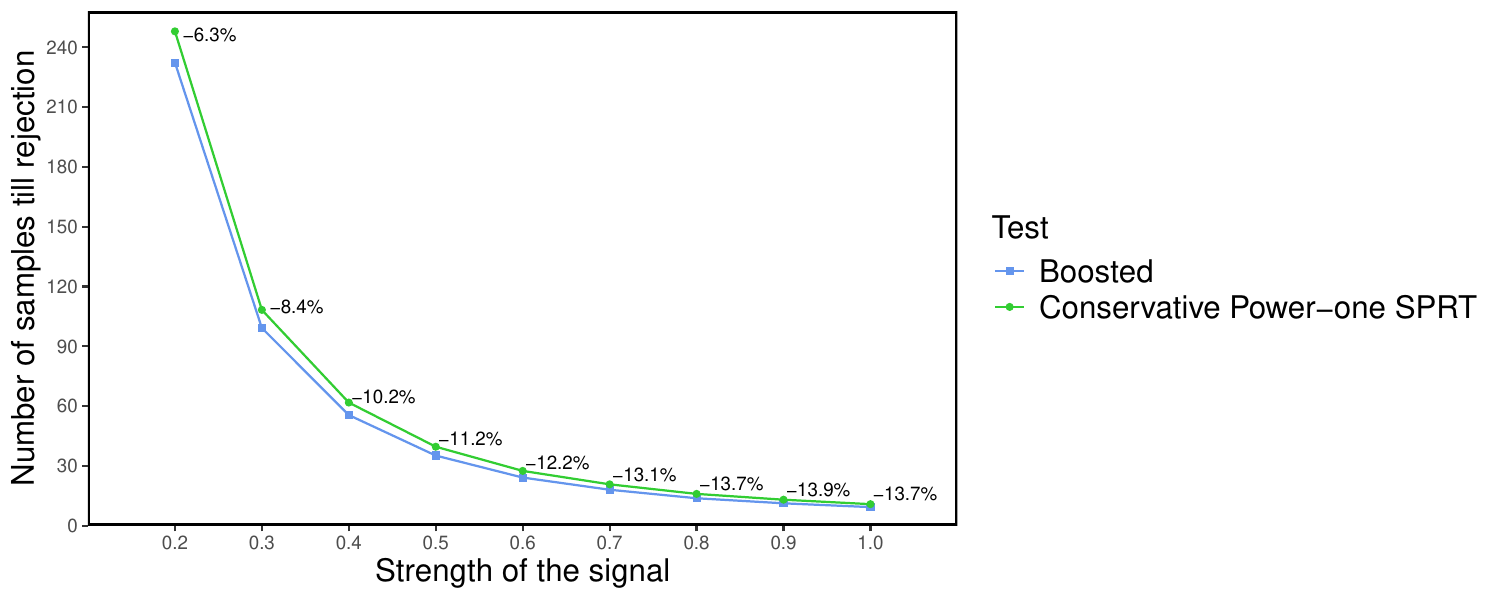}
\caption{Comparison of sample size between the SPRT with a predictable plugin and the boosted SPRT in a Gaussian testing setup with composite alternative for different signal strengths. The boosted SPRT saves $6.3-13.9\%$ of sample size compared to the SPRT.  \label{fig:sim_plugin} }\end{figure}

\subsection{Handling composite nulls in one-sided testing problems\label{sec:comp_null}}

An even more common testing setup in practice is
\begin{align}
H_0: \theta^* \leq \theta_0 \quad \text{vs.} \quad H_1:\theta^* >\theta_0 \quad (\theta_0\in \Theta), \label{eq:comp_null}
\end{align}
meaning we have a composite null hypothesis and a composite alternative. Throughout this section, we assume that for all $\theta', \theta''\in \Theta$ with $\theta''\geq \theta'$,
\begin{align}
\lambda(x)=\frac{p_{\theta''}(x)}{p_{\theta'}(x)} \text{ is nondecreasing in } x.\label{eq:monotone_LR}
\end{align}
\citet{robbins1973class} showed that under this assumption the process $(\Lambda_t^{\textrm{plugin}})$ defined in Equation~\eqref{eq:Z_t} also provides a valid sequential test for the composite null hypothesis above. In the following, we show that the same holds for our boosted process.

\begin{proposition}\label{prop:one_sided_null}
    Consider the one-sided testing problem \eqref{eq:comp_null} under the monotone likelihood ratio property \eqref{eq:monotone_LR} and let $\lambda_t(X_t)=\frac{p_{\theta_t}(X_t)}{p_{\theta_0}(X_t)}$ for some predictable parameter $\theta_t\geq \theta_0$. Then,  
    $$
    \Ex_{\theta_0}[T_{\alpha}(b_t\lambda_t(X_t);M)|\Lambda_{t-1}^{\boost}]\leq 1 \implies \Ex_{\theta}[T_{\alpha}(b_t\lambda_t(X_t);M)|\Lambda_{t-1}^{\boost}]\leq 1 \quad (\theta\leq \theta_0)
    $$
    for any $\alpha\in (0,1)$, $M>0$ and $b_t\geq 1$.
\end{proposition}


To summarize, for the composite testing problem above and under the monotone likelihood ratio property, we can do exactly the same as described in Section \ref{sec:comp_alt} for the simple null hypothesis without inflating type I error.

\begin{remark}
    An important class satisfying property \eqref{eq:monotone_LR} is the one-parameter exponential family, where each density is given by
$
p_{\theta}(x)=h(x)\exp(T(x)\eta(\theta)-A(\theta))
$
with $h$, $T$, $\eta$ and $A$ being known functions.
It is easy to see that \eqref{eq:monotone_LR} is fulfilled for the sufficient statistic $T(x)$, if $\eta(\theta)$ is nondecreasing in $\theta$.
\end{remark}

\section{Nontrivial extensions and applications\label{sec:extensions}}

We now demonstrate some nontrivial extensions and applications of our boosting method. First, we show how it can be used for tighten confidence sequences in Section~\ref{sec:CS}. Then, we exemplify its use in sampling without replacement situations (Section~\ref{sec:WoR}) and with conformal martingales (Section~\ref{sec:conformal_martingales}) for testing the assumption of i.i.d. data.

\subsection{Confidence sequences\label{sec:CS}}

A $(1-\alpha)$-level confidence sequence (CS) for $\theta^*\in \Theta$ is a sequence of sets $(C_t)_{t\in \mathbb{N}}$ such that
\begin{align}\mathbb{P}_{\theta^*}(\exists t\in \mathbb{N}:\theta^*\notin C_t)\leq \alpha.\label{eq:CS_def}\end{align}
Let $\tau_\theta$, $\theta\in \Theta$, be a sequential test for $H_0^{\theta}:\theta^*=\theta$ with property \eqref{eq:power-one_0}. Then 
$$
C_t=\{\theta\in \Theta: \tau_\theta>t\}
$$
is a CS for $\theta^*$. Let $(M_t^{\theta})$ be the test supermartingale for $H_0^{\theta}$ that defines  $\tau_\theta$, meaning $\tau_\theta=\inf\{t\in \mathbb{N}: M_t^{\theta}\geq 1/\alpha\}$. It follows immediately that $(C_t)_{t\in \mathbb{N}}$ can be boosted, meaning each $C_t$ shrinks in size while maintaining \eqref{eq:CS_def}, by boosting each of the martingales  $(M_t^{\theta})$, $\theta\in \Theta$, by Algorithm~\ref{alg:general}. However, the computational effort of comparing $M_t^{\theta}$ with $1/\alpha$ for every $\theta\in \Theta$ is extremely high or even infeasible. Due to computational convenience and interpretability, in practice one usually focuses on  CSs that form an interval
$
C_t=(l_t,u_t).
$
Often, two-sided CSs can then be obtained as the intersection of an upper and a lower $(1-\alpha/2)$-level CS. Hence, for now, let us for look a sequence $(l_t)_{t\in \mathbb{N}}$ such that   
$$\mathbb{P}_{\theta^*}(\exists t\in \mathbb{N}:\theta^*< l_t)\leq \alpha.$$
In order to ensure that the resulting CS has this form, it is required that the martingales $(M_t^{\theta})_{\theta \in \Theta}$ are nonincreasing in $\theta$ for every $t\in \mathbb{N}$. Now note that each boosting factor $b_t^{\theta}$, $\theta\in \Theta$ and $t\in \mathbb{N}$, defined by Algorithm~\ref{alg:general}, is a nondecreasing function of the (boosted) supermartingale at the previous step $M_{t-1}^{\theta, \boost}$. Therefore, the boosting factors $(b_t^{\theta})_{\theta\in \Theta}$ are also nonincreasing in $\theta$, provided that the distribution of $(M_t^{\theta})$ under $H_0^{\theta}$ is the same for all $\theta\in \Theta$. This implies that also the boosted processes $(M_t^{\theta, \mathrm{boost}})_{\theta\in \Theta}$ are nonincreasing in $\theta$ for each $t$ such that the resulting CS provides indeed a lower bound. With this, the boosted CS can be easily determined computationally by setting $l_t$ at each step $t\in \mathbb{N}$ to be the smallest $\theta\in \Theta$ such that $M_t^{\theta,\boost}< 1/\alpha$. In Supplementary Material~\ref{sec:example_CS}, we illustrate this by boosting a CS for Gaussian means.

\subsection{Sampling without replacement\label{sec:WoR}}
Suppose we have some finite population $\{x_1,\ldots,x_N\}$ and draw samples
$$
X_t|\{X_1,\ldots,X_{t-1}\} \sim \text{Uniform}(\{x_1,\ldots,x_N\}\setminus \{X_1,\ldots, X_{t-1}\})
$$
sequentially and without replacement with the goal of testing the hypothesis 
$$H_0:\mu^*\leq \mu_0 \quad \text{vs.} \quad H_1:\mu^*>\mu_0,$$
where $\mu^*=\frac{1}{N} \sum_{i=1}^N x_i$, because the effort of drawing all samples is too high. In this section, we focus on binary variables $X_i$. For example, this scenario is encountered in Risk-limiting audits (RLA) used for auditing the results of an election \citep{waudby2021rilacs} or when checking whether the permutation p-value is below the predefined level $\alpha$ \citep{waudby2020confidence}. Also note that this testing scenario is very different from the classical i.i.d. Bernoulli SPRT setting, since the conditional mean $C_i(\mu_0)= \frac{N\mu_0-\sum_{j=1}^{i-1} X_j}{N-i+1}$ of $X_i|X_1,\ldots, X_{i-1}$, if the mean of $\{x_1,\ldots,x_N\}$ is $\mu_0$, changes after each sample. In the following, we show how the state-of-the-art method for this setting can be improved by our boosting technique.



\citet{waudby2021rilacs} proposed the RiLACS test martingale
\begin{align}
M_t=\prod_{i=1}^t (1+\lambda_i(X_i-C_i(\mu_0))), \label{eq:mart_WoR}
\end{align}
where $\lambda_i\in \left[0,\frac{1}{C_i(\mu_0)}\right]$ is a hyperparameter measurable with respect to $\mathcal{F}_{i-1}$. Hence, each individual factor of $(M_t)$ is just given by $L_t=1+\lambda_i(X_i-C_i(\mu_0))$. Under the null hypothesis and conditional on $\mathcal{F}_{t-1}$, we have that $L_t=1+\lambda_i(1-C_i(\mu_0))$ with probability at most $C_i(\mu_0)$ and $L_t=1-\lambda_iC_i(\mu_0)$ with probability at least $1-C_i(\mu_0)$.  This implies that we can determine exact boosting factors $b_t$ by
\begin{align*} b_t=\begin{cases} \frac{1-C_i(\mu_0)/(M_{t-1}^{\boost}\alpha)}{(1-C_i(\mu_0))(1-\lambda_iC_i(\mu_0))}, & \text{ if } M_{t-1}^{\boost} (1+\lambda_i(1-C_i(\mu_0))) \geq 1/\alpha\\
1, & \text{ otherwise.}
\end{cases}\end{align*}

 \citet{waudby2021rilacs} proposed to choose $\lambda_i=\min (1/C_i(\mu_0), 2(2\mu_1-1))$, if the alternative $H_1:\mu^*=\mu_1$  is known in advance. For example, this is the case in Risk-limiting audits. In Figure~\ref{fig:sim_WoR}, we compare the required sample size of RiLACS and its boosted improvement for $\mu_0=0.5$ and $\mu_1=0.55$ for different total number of samples $N$ at level $\alpha=0.01$. The gain in sample size obtained by boosting is small but visible. This is not surprising since the RiLACS test martingale \eqref{eq:mart_WoR} only overshoots slightly at $1/\alpha$ due to the binary data. However, since $b_t$ has a closed form, this improvement is essentially free.

\begin{figure}[h!]
\centering
\includegraphics[width=0.8\textwidth]{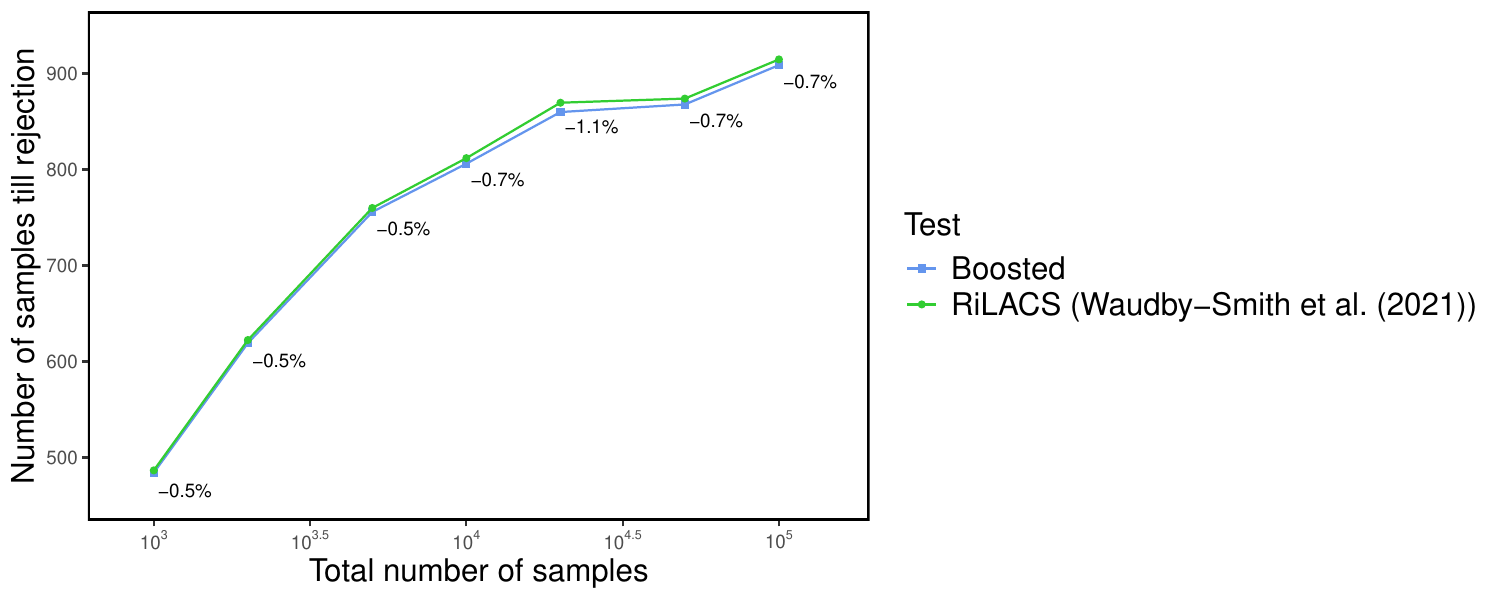}
\caption{Comparison of the required sample size by RILACS \citep{waudby2021rilacs} and its boosted improvement in a sampling without replacement setting for different total number of samples. Boosting leads to a small but consistent gain in sample size. 
\label{fig:sim_WoR} }\end{figure}

\subsection{Conformal martingales\label{sec:conformal_martingales}}

A common assumption in statistics has been that the data is generated independently from the same distribution, also called i.i.d. (independent and identically distributed). Here, we consider testing this assumption sequentially, meaning to test the composite null hypothesis 
$$
H_0:X_1,X_2,\ldots \text{ are i.i.d.}
$$
From a single data sequence it is impossible to distinguish whether the data is i.i.d. or exchangeable \citep{ramdas2022testing}, where $(X_t)_{t\in \mathbb{N}}$ is called exchangeable if $(X_1,\ldots,X_t)$ has the same joint distribution as $(X_{\sigma(1)},\ldots,X_{\sigma(t)})$ for every $t\in \mathbb{N}$ and permutation $\sigma$ of the first $t$ indices. Hence, $H_0$ can be formulated equivalently as 
$$
H_0:X_1,X_2,\ldots \text{ are exchangeable.}
$$
Vovk's conformal martingale approach \citep{vovk2021testing} allows to reduce $H_0$ to a simple null hypothesis. This makes it not only possible to test $H_0$ sequentially using a test martingale, but also to apply our boosting technique from Algorithm~\ref{alg:general}. We first recap the method by Vovk and then demonstrate how we can apply our boosting approach on top of it.

Let $A$ be a \textit{nonconformity measure}, which maps observations $(x_1,\ldots, x_t)$ of any length to a sequence of numbers $(z_1,\ldots,z_t)\in \mathbb{R}^t$ such that for any $t$ and any permutation $\sigma$,
$$
A(x_1,\ldots,x_t)=(z_1,\ldots, z_t)\implies A(x_{\sigma(1)},\ldots,x_{\sigma(t)})=(z_{\sigma(1)},\ldots,z_{\sigma(t)}).
$$
With this, for each $t\in \mathbb{N}$, one can calculate a \enquote{conformal p-value} by
$$
p_t=\frac{|\{i\leq t:z_i>z_t\}|+\theta_t|\{i\leq t: z_i=z_t\}|}{t},
$$
where $(z_1,\ldots, z_t)=A(x_1,\ldots, x_t)$ and $\theta_t$ is randomly and independently sampled from $U[0,1]$. Vovk showed that the p-values $p_1, p_2,\ldots$ are i.i.d. from $U[0,1]$ under the null hypothesis $H_0$, and proposed to collect evidence against $H_0$ by forming a \textit{conformal martingale}
$$
M_t=\prod_{i=1}^t f_i(p_i),
$$
where the `betting function' $f_i:[0,1]\rightarrow [0,\infty]$, $i\in \mathbb{N}$ may depend on the previous p-values and satisfies $ \int_0^1 f_i(u)du=1$. It is easy to check that this ensures $\mathbb{E}_0[f_t(p_t)|\mathcal{F}_{t-1}]=1$, where $\mathcal{F}_{t-1}=\sigma(p_1,\ldots,p_{t-1})$, such that $(M_t)$ is indeed a test martingale for $H_0$ and $f_i(p_i)$  are the corresponding individual factors. Note that $(M_t)$ may not be a martingale with respect to the filtration generated by the observations $X_1,X_2,\ldots$, but it surely is with respect to the filtration generated by the p-values $p_1,p_2,\ldots$ . Since we know that the p-values $p_1,p_2,\ldots$ are i.i.d. uniformly distributed on $U[0,1]$ under $H_0$, we can easily determine the distribution of $f_i(p_i)$ conditional on the previous p-values such that Algorithm~\ref{alg:general} can be applied.

For example, a concrete betting function proposed by Vovk is
$
f^{\kappa}(u)=\kappa u^{\kappa-1}, \quad \kappa\in (0,1).
$
 At each step $t\in \mathbb{N}$, setting $\mathbb{E}_0[T_{\alpha}(b_tf_t(p_t);M_{t-1}^{\boost})|\mathcal{F}_{t-1}]=1$ yields the equation
$$
b_t \max\left(1-\left(\frac{1/(\alpha M_{t-1}^{\boost})}{b_t \kappa}\right)^{\frac{\kappa}{\kappa -1}}, 0\right)+\frac{1}{\alpha M_{t-1}^{\boost}} \min\left(\left(\frac{1/(\alpha M_{t-1}^{\boost})}{b_t \kappa}\right)^{\frac{\kappa}{\kappa -1}},1\right)=1,
$$
which can be solved for $b_t$. For example, $\alpha=0.05$, $M_{t-1}^{\boost}=1$ and $\kappa=0.5$ lead to a boosting factor of $b_t=1.013$. The boosting factor is decreasing in $\kappa$ such that we obtain a boosting factor of $b_t=2.199$ for $\kappa=0.1$. In practice, it can be useful to choose a predictable sequence of parameters $(\kappa_t)_{t\in \mathbb{N}}$ and employ the betting function $f^{k_t}$ for the p-value $p_t$ \citep{vovk2021testing}. However, this poses no problem for our boosting algorithm (see Section~\ref{sec:comp_alt}).

\section{
Challenging Wald's two-sided approximate SPRT
\label{sec:futility}}

While the previous sections focused on power-one SPRTs ($\beta=0$) where we only stop to reject the null but not to accept it, we now consider the case of $\beta > 0$, where we also wish to stop if it is unlikely that we would reject the null if we continued testing. 
We will show how we can use a boosting technique similar to that in Section~\ref{sec:avoid_overshoot} to construct powerful anytime-valid tests that allow to ``stop for futility'' (stop and accept the null).

Our main comparison point will be Wald's approximate two-sided SPRT, which uses thresholds of $(1-\beta)/\alpha$ and $\beta/(1-\alpha)$. Importantly, out of the three desiderata of (a) type I error control at $\alpha$, (b) type II error control at $\beta$, (c) optimal sample size, these approximate thresholds guarantee none of the three properties. In this section, we will first develop a general method with valid type I error control for test supermartingales, and then show how the stop for futility can be chosen to obtain tight type I and type II error control in the SPRT case. The latter method often outperforms Wald's approximate SPRT in the sense that while guaranteeing (a) and (b), our method also manages to typically/often stop sooner than Wald's (though not guaranteeing (c)).

\subsection{\smash{Two-sided stopping with type I error control for supermartingales}\label{sec:two_sided_typeI}}

Suppose our test is based on a test (super)martingale $(M_t)$, such as the likelihood ratio. We consider stopping for futility if $M_t\leq \nu_t$, where $\nu_t\leq 1/\alpha$ is a predictable parameter, meaning $\nu_t$ is measurable with respect to $\mathcal{F}_{t-1}$. For example, Wald's approximation suggests to set $\nu_t=\frac{\beta}{1-\alpha}$. Note that stopping for futility can never increase the type I error, but it may hurt the power. Below, we will show how we can gain back part of that power.

If $M_t>0$, we have $\mathbb{E}_0[M_{t'}|\mathcal{F}_t]=M_t>0$ for all $t'>t$. Hence, if we stop when our test martingale  is strictly positive, we make a conservative decision since there is a positive probability of rejecting if we continue sampling. Therefore, our approach is to form a truncated test (super)martingale $(M_t^{\mathrm{trunc}})$ such that $M_t\leq \nu_t$ implies that $M_t^{\mathrm{trunc}}=0$. To achieve this, we extend our truncation function $T_{\alpha}(x; M)$ from \eqref{eq:truncation} as follows:
\begin{align}
T_{\alpha}(x; M, \nu)=\begin{cases}
        0, &\text{ if } Mx\leq \nu \\
        x, &\text{ if } \nu<Mx\leq \frac{1}{\alpha} \\
        \frac{1}{M \alpha}, &\text{ if } Mx> \frac{1}{\alpha}.
    \end{cases} 
    \label{eq:truncation_futility}
\end{align}
We can use $T_{\alpha}(x; M, \nu)$ for boosting in the same manner as we did with $T_{\alpha}(x; M)$ for the one-sided stopping boundary by choosing a boosting factor $b_t\geq 1$ such that
\begin{align}\Ex_{0}[T_{\alpha}(b_t L_t;M_{t-1}^{\mathrm{boost}}, \nu_t)|\mathcal{F}_{t-1}]\leq 1,\label{eq:boosting_inequality_2side}\end{align}
where $M_{t-1}^{\mathrm{boost}}=\prod_{i=1}^{t-1} T_{\alpha}(b_iL_i; M_{i-1}^{\mathrm{boost}}, \nu_i) $. Note that $T_{\alpha}(x; M, \nu)=T_{\alpha}(x; M)$, if $\nu=0$. Furthermore, $T_{\alpha}(x; M, \nu)\leq T_{\alpha}(x; M)$ for all $x\in \mathbb{R}_{\geq 0}$, implying that the boosting factors obtained by  $T_{\alpha}(x; M, \nu)$ are larger than the boosting factors computed with $T_{\alpha}(x; M)$. The following theorem extends Theorem~\ref{theo:main} by including a specific stop for futility.
\begin{theorem}\label{theo:main_2side}
    Let $\nu_t^M\leq 1/\alpha$, $t\in \mathbb{N}$, be a predictable parameter and $\nu_t=\min\left(\nu_t^M\prod_{i=1}^t b_i , 1/\alpha\right)$. Furthermore, let $(M_t)$ be any test (super)martingale and $(M_t^{\mathrm{boost}})$ be the process obtained by applying Algorithm~\ref{alg:general} with $T_{\alpha}(x; M_{t-1}^{\boost}, \nu_t)$ and define the sequential tests $\delta_{M}=\mathbbm{1}\{M_{\tau_{M}}\geq 1/\alpha\}$ and $\delta_{\mathrm{boost}}=\mathbbm{1}\{M_{\tau_{\mathrm{boost}}}^{\mathrm{boost}}\geq 1/\alpha\}$, where $\tau_{M}=\inf\{t\in \mathbb{N}: M_t\geq 1/\alpha \lor M_t\leq \nu_t^M\}$ and $\tau_{\mathrm{boost}}=\inf\{t\in \mathbb{N}: M_t^{\mathrm{boost}}\geq 1/\alpha \lor M_t^{\mathrm{boost}}\leq \nu_t\}$. Then $M_t^{\mathrm{boost}}\geq M_t$ for all $t< \tau_{\boost}$,
    $\delta_{\mathrm{boost}}\geq \delta$ and $\tau_{\mathrm{boost}}\leq \tau_M$. Furthermore, $(M_t^{\boost})$ is a test supermartingale. Therefore, 
    $\mathbb{P}_0(\delta_{\mathrm{boost}}=1)\leq \alpha$ and $\mathbb{E}_0[M_{\tau}^{\boost}]\leq 1$ for all stopping times $\tau$. 
\end{theorem}

Note that in the above theorem $\nu_t$ is a function of $b_t$. This ensures that our boosted martingale $(M_t^{\boost})$ always stops in cases where the initial martingale $(M_t)$ stops for futility, which allows to prove that $\tau_{\mathrm{boost}}\leq \tau_M$. However, we think in practice both scenarios are reasonable. We could either want to provide a stop for futility $\nu_t$ directly (independently of $b_t$) for $M_t^{\boost}$ or provide a stop for futility $\nu_t^M$ for $M_t$ (independently of $b_t$) and then set $\nu_t=\min\left(\nu_t^M\prod_{i=1}^t b_i , 1/\alpha\right)$. For the rest of this subsection, we treat $\nu_t$ as independent of $b_t$. However, all results shown apply equivalently for $\nu_t=\min\left(\nu_t^M\prod_{i=1}^t b_i , 1/\alpha\right)$, where $\nu_t^M$ does not dependent on $b_t$.

It turns out that in the two-sided case Proposition~\ref{prop:solvability} no longer holds, meaning we cannot always find a $b_t$ which solves \eqref{eq:boosting_inequality_2side} with an equality.
\begin{example}\label{example:binary}
     Suppose we have i.i.d. binary data $X_1,X_2,\ldots$ and $\mathbb{E}_0[X_t]=1/3$ and $\mathbb{E}_1[X_t]=2/3$. Then the likelihood ratio $\lambda(X_t)=2$, if $X_t=1$, and $\lambda(X_t)=1/2$ otherwise. If $\alpha=0.05$, $\Lambda_{t-1}^{\boost}=8$ and $\nu_t=7$. Then $\Ex_{0}[T_{\alpha}(b_t L_t;\Lambda_{t-1}^{\mathrm{boost}}, \nu_t)|\mathcal{F}_{t-1}]\leq 1/3\cdot 20/8<1$ for all $1\leq b_t\leq 7/4$ and $\Ex_{0}[T_{\alpha}(b_t L_t;\Lambda_{t-1}^{\mathrm{boost}}, \nu_t)|\mathcal{F}_{t-1}]\geq 2/3\cdot 7/8 +1/3\cdot 20/8>1$ for all $b_t>7/4$.
\end{example}
  The problem in Example~\ref{example:binary} is that $T_{\alpha}(x;M,\nu)$ is not continuous in $x$ and that the data is binary, creating a jump in the expectation of the boosted (super)martingale factor. In Supplementary Material~\ref{sec:boosting_discrete}, we discuss how external randomization could be used to avoid conservatism with discrete data when using a stop for futility and in the following proposition we show that Proposition~\ref{prop:solvability} can at least be extended to two-sided stopping for continuously distributed martingale factors.

\begin{proposition}\label{prop:solvability_2side}
    If $L_t=M_t/M_{t-1}$ is continuously distributed under $\mathbb{P}_0$
    and ${M_{t-1}^{\boost}\in(0,\alpha^{-1})}$, then there exists a $b_t^*\geq 1$ such that 
    \begin{align}\Ex_{0}[T_{\alpha}(b_t^* L_t;M_{t-1}^{\mathrm{boost}}, \nu_t)|\mathcal{F}_{t-1}]= 1.\label{eq:boosting_equality_2side}\end{align}
    Further, for every other $b_t\geq 1$ that satisfies \eqref{eq:boosting_inequality_2side}, we have that $T_{\alpha}(b_t^* L_t;M_{t-1}^{\mathrm{boost}}, \nu_t)\geq T_{\alpha}(b_t L_t;M_{t-1}^{\mathrm{boost}}, \nu_t)$ $\mathbb{P}_0$-almost surely.
\end{proposition}

In the same manner as in \eqref{eq:boost_SPRT}, we can calculate for SPRTs:
\begin{align*}
    &\mathbb{E}_{\theta_0}[T_{\alpha}(b_t \lambda(X_t) ; \Lambda_{t-1}^{\boost}, \nu_t)| \Lambda_{t-1}^{\boost}] \\ &= b_t\mathbb{P}_{\theta_1}\left(\frac{\nu_t}{b_t \Lambda_{t-1}^{\boost}}<\lambda(X_t)\leq \frac{1}{b_t\alpha \Lambda_{t-1}^{\boost}}\bigg\vert \Lambda_{t-1}^{\boost}\right)  +\frac{1}{\Lambda_{t-1}^{\boost}\alpha}\mathbb{P}_{\theta_0}\left(\lambda(X_t)> \frac{1}{b_t\alpha \Lambda_{t-1}^{\boost}}\bigg\vert \Lambda_{t-1}^{\boost}\right).
\end{align*}

The boosting factors obtained in the same setup as for Table~\ref{tab:boosting_factors}, but with a stop for futility at $\nu_t=0.4$, is shown in Table~\ref{tab:boosting_factors_futil}. Obviously, all boosting factors are larger when the stop for futility is included. Furthermore, especially in cases where the current martingale value $\Lambda_{t-1}^{\boost}$ is small, including the stop for futility increases the boosting factor.

\begin{table}[!htb]
    \centering
    \begin{tabular}{l|c c c c c }
        \diagbox[]{$\delta$}{$\Lambda_{t-1}^{\boost}$} & $0.5$ & $1$ & $2$ & $4$ & $10$ \\ \hline
        $0.1$ & $1.00895$ & $1$ & $1$ & $1$ & $1.00001$ \\ 
        $0.5$ & $1.17964$ & $1.01743$ & $1.00026$ & $1.00019$ & $1.03019$ \\
        $1$ & $1.21801$ & $1.07547$ & $1.02817$ & $1.05651$ & $1.37357$ \\ 
        $2$ & $1.32013$ & $1.38991$ & $1.62094$ & $2.21769$ & $5.76214$ \\ 
        $3$ & $2.73073$ & $3.75762$ & $6.00201$ & $12.1467$ & $68.8295$ \\ 
    \end{tabular}
    \caption{Boosting factors obtained for $\alpha=0.05$ and $\nu_t=0.4$, $t\in \mathbb{N}$, in a simple Gaussian testing problem with mean difference $\delta$ and variance $1$. \label{tab:boosting_factors_futil}}
\end{table}

\begin{remark}
    Note that our boosted SPRT could also be written as
    $$
    \text{reject }H_0, \text{ if } \Lambda_t\geq \gamma_1^t, \quad  \text{accept } H_0, \text{ if }  \Lambda_t\leq \gamma_0^t, \quad \text{and continue sampling if } \gamma_0^t < \Lambda_t< \gamma_1^t,
    $$
    where $\gamma_1^t=\frac{1}{\alpha\prod_{i=1}^t b_i}$ and $\gamma_0^t=\frac{\nu_t}{\prod_{i=1}^t b_i}$ are predictable parameters. Such generalized SPRTs with non-constant boundaries over time were already introduced by \citet{weiss1953testing}, however, Weiss assumed that $\gamma_0^t$ and $\gamma_1^t$, $t\in \mathbb{N}$,  are prespecified constants.
\end{remark}

\begin{remark}
    Typically, the parameter $\nu_t$ is set to some value smaller than $1$, as $\Lambda_t<1$ implies that the current data favors the null hypothesis over the alternative. However, $\nu_t$ is a predictable parameter and one could also think of choosing it based on the current value of $\Lambda_{t-1}^{\boost}$. For example, it might be reasonable to set $\nu_t=8$ if $\Lambda_{t-1}^{\boost}=10$ to stop in case of decreasing evidence. In general, it is even possible to set $\nu_t>\Lambda_{t-1}^{\boost}$. 
\end{remark}

\subsection{Two-sided stopping with type I, II error control for SPRTs}

We know that boosting always ensures type I error control at the level $\alpha$ by Ville's inequality \eqref{eq:ville}. Now we want to use the same approach to choose the stop for futility $\nu_t$ in a way that ensures valid type II error control at a level $\beta\in (0,1)$. The idea is to use a second SPRT for testing the alternative $H_1:X_i\sim \mathbb{P}_1$ against the null hypothesis $H_0:X_i\sim \mathbb{P}_0$ and stop for futility with our original test if this \enquote{inverse SPRT} rejects $H_1$. Hence, instead of using the likelihood ratio process $(\Lambda_t)$ of $\mathbb{P}_1$ against $\mathbb{P}_0$, we use its inverse $\Lambda_t^{\mathrm{inv}}=1/\Lambda_t$. By Ville's inequality \eqref{eq:ville}, we control the type II error for the original test of $H_0$ against $H_1$ at level $\beta$, if we stop for futility when the inverse likelihood ratio $\Lambda_t^{\mathrm{inv}}$ crosses $1/\beta$. However, as we demonstrated in the previous sections, Ville's inequality is loose if there is overshoot. Hence, to get tighter type II error control, we need to boost the inverse SPRT in the same manner as we did for the usual SPRT and include the stop for futility yielded by the inverse SPRT in the truncation function  \eqref{eq:truncation_futility}. Precisely, denoting $\lambda_t(X_t)=p_1(X_t)/p_0(X_t)$ at each step $t\in \mathbb{N}$, we choose $b_t\geq 1$ and $b_t^{\mathrm{inv}}\geq 1$ as large as possible such that
\begin{align}
    \mathbb{E}_0\left[T_{\alpha}(b_t \lambda_t(X_t) ; \Lambda_{t-1}^{\boost}, \nu_t)\Big\vert \mathcal{F}_{t-1} \right] \leq 1,  &\text{ where } \nu_t=\min(\beta \prod_{i=1}^t b_i \prod_{i=1}^t b_i^{\mathrm{inv}}, 1/\alpha) ,\label{eq:boost_original} \\
    \mathbb{E}_1\left[T_{\beta}(b_t^{\mathrm{inv}}/\lambda_t(X_t) ; \Lambda_{t-1}^{\boost, \mathrm{inv}}, \nu_t^{ \mathrm{inv}})\Big\vert \mathcal{F}_{t-1} \right] \leq 1,  &\text{ where } \nu_t^{ \mathrm{inv}}=\min(\alpha \prod_{i=1}^t b_i \prod_{i=1}^t b_i^{\mathrm{inv}}, 1/\beta),\label{eq:boost_inverse}
\end{align}
$$ \Lambda_t^{\boost}= \prod_{i=1}^t T_{\alpha}(b_i \lambda_i(X_i) ; \Lambda_{i-1}^{\boost}, \nu_i)\quad \text{and} \quad \Lambda_t^{\boost, \mathrm{inv}}=\prod_{i=1}^t T_{\alpha}(b_i^{\boost} /\lambda_i(X_i) ; \Lambda_{i-1}^{\boost, \mathrm{inv}}, \nu_i^{\mathrm{inv}}).$$ The definition of $\nu_t$ ensures that the original SPRT only stops for futility when the boosted inverse SPRT stops for rejection. We next use this to yield valid error control.

\begin{theorem}\label{theo:boosting_type_I_II}
    Let the boosting factors $b_1, b_2,\ldots$ and inverse boosting factors $b_1^{\mathrm{inv}}, b_2^{\mathrm{inv}},\ldots$ be chosen such that \eqref{eq:boost_original} and \eqref{eq:boost_inverse} are fulfilled for each $t\in \mathbb{N}$. Then the boosted SPRT given by $\delta^{\boost}=\mathbbm{1}\{\Lambda_{\tau_{\boost}}^{\boost}\geq 1/\alpha\}$ with $\tau_{\boost}=\inf\{t\in \mathbb{N}: \Lambda_{t}^{\boost}\geq 1/\alpha \lor \Lambda_{t}^{\boost}\leq \nu_t\}$ controls the type I error at level $\alpha$ and the type II error at level $\beta$.
\end{theorem}

Similar to before, it follows that the system above can often be solved with equalities.

\begin{proposition}\label{prop:boosting_equal_2side}
    If $\lambda_t(X_t)$ is continuously distributed under $\mathbb{P}_0$ and $\Lambda_{t-1}^{\boost}\in (0,\alpha^{-1}), \Lambda_{t-1}^{\boost, \textrm{inv}} \in (0,\beta^{-1})$, then there exist $b_t^*$ and $b_t^{*,\textrm{inv}}$ satisfying \eqref{eq:boost_original} and \eqref{eq:boost_inverse} with equality.    
\end{proposition}

Note that even if \eqref{eq:boost_original} and \eqref{eq:boost_inverse} are satisfied with an equality, the error control can still be slightly conservative, since it may happen that 
$$
\beta \prod_{i=1}^t b_i \prod_{i=1}^t b_i^{\mathrm{inv}} > 1/\alpha, \ \text{ or equivalently, } \ \alpha \prod_{i=1}^t b_i \prod_{i=1}^t b_i^{\mathrm{inv}} > 1/\beta,
$$
which means that the rejection region of the boosted SPRT and the boosted inverse SPRT overlap. This is why we set $\nu_t=1/\alpha$ and $\nu_t^{\mathrm{inv}}=1/\beta$ in this case which ensures valid type I and type II error control but avoids the simultaneous acceptance and rejection of both hypotheses (in general, we can choose smaller $\nu_t$ and $\nu_t^{\textrm{inv}}$ than given on the right-hand side of Equations~\eqref{eq:boost_original} and \eqref{eq:boost_inverse} without violating error control). However, this conservatism is expected to be small, particularly if $\alpha$ and $\beta$ are small.  Next, we show that the boosted two-sided SPRT always beats Wald's conservative SPRT in terms of sample size.

\begin{proposition}\label{prop:improve_cons_SPRT}
Let the boosting factors $b_1, b_2,\ldots$ and inverse boosting factors $b_1^{\mathrm{inv}}, b_2^{\mathrm{inv}},\ldots$ be chosen such that \eqref{eq:boost_original} and \eqref{eq:boost_inverse} are fulfilled for each $t\in \mathbb{N}$. Furthermore, let $\tau_{\boost}=\inf\{t\in \mathbb{N}: \Lambda_{t}^{\boost}\geq 1/\alpha \lor \Lambda_{t}^{\boost}\leq \nu_t\}$ be stopping time of the boosted SPRT and $\tau_{\textrm{cons}}=\inf\{t\in \mathbb{N}: \Lambda_{t}\geq 1/\alpha \lor \Lambda_{t}\leq \beta\}$ of Wald's conservative SPRT. Then $\tau_{\boost}\leq \tau_{\textrm{cons}}$.
\end{proposition}

\begin{figure}[h!]
\centering
\includegraphics[width=\textwidth]{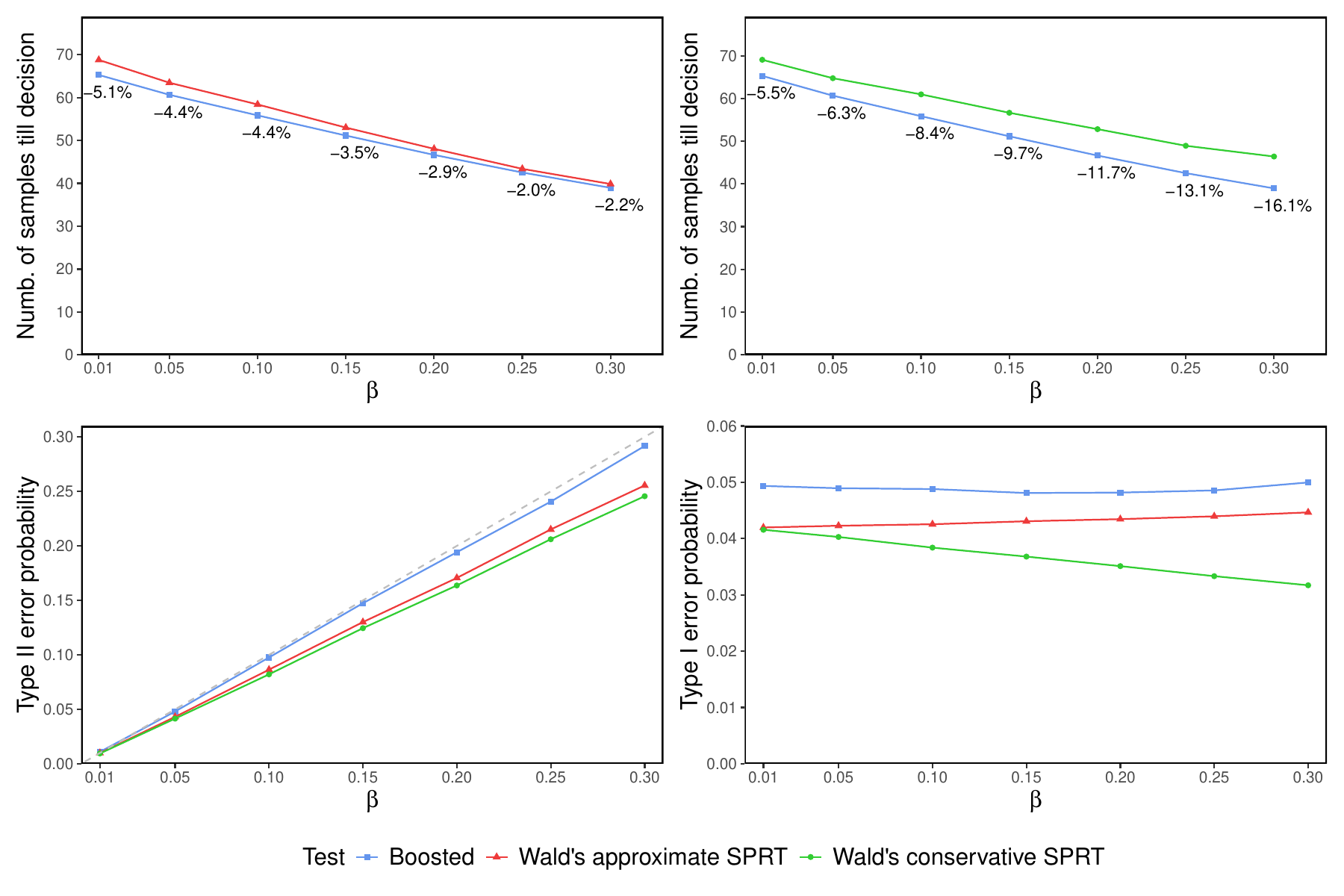}
\caption{Comparison of sample size, type II and type I error probability between  Wald's approximate SPRT, Wald's conservative SPRT and the boosted SPRT in a simple Gaussian testing setup for different desired type II error probabilities ($\beta$)  and $\alpha=0.05$. Wald's conservative SPRT and our boosted method provably control the type I and type II error, while the SPRT with Wald's approximated thresholds does not. Furthermore, the boosted SPRT saves sample size and nearly exhausts the error probabilities. \label{fig:sim_futility} }\end{figure}

Note that it is not possible to generally improve Wald's approximate SPRT, since it does not provide valid type I and type II error control. However, in Figure~\ref{fig:sim_futility} we compare the sample size, the type II and the type I error probability of the SPRT with Wald's approximated and conservative thresholds to our boosted method for different parameters $\beta$ in the same Gaussian simulation setup as considered in Section~\ref{sec:simple_null_alt}. The desired type I error probability was set to $\alpha=0.05$ and the strength of the signal to $\delta=0.3$. The boosted method reduced the required sample size in all cases. Furthermore, the boosted SPRT nearly exhausts the type I error and type II error probability, while even Wald's approximate SPRT is conservative in this case. For all $\beta$, one boosted SPRT, including the calculation of the boosting factors, needed less than $0.3$ seconds to compile on average. In Supplementary Material~\ref{appn:sims} we repeat this experiment for $\alpha=0.01$ instead of $\alpha=0.05$, which leads to similar results. In Supplementary Material~\ref{sec:implementation_2side}, we give more details on our implementation of this boosted two-sided SPRT.

\begin{remark}
    In order to avoid the simultaneous solving for $b_t$ and $b_t^{\mathrm{inv}}$, one could also just set $\nu_t'=\min\left(\beta \prod_{i=1}^t b_i, 1/\alpha\right)$ and choose $b_t\geq 1$ as large as possible such that 
    $$
    \mathbb{E}_0\left[T_{\alpha}(b_t \lambda_t(X_t) ; \Lambda_{t-1}^{\boost}, \nu_t')\Big\vert \mathcal{F}_{t-1} \right] \leq 1.
    $$
    Since $\nu_t'\leq \nu_t$ for $\nu_t$ in \eqref{eq:boost_original}, this would lead to valid type I and type II error control. However, the type II error control would usually be more conservative than with the proposed approach.
\end{remark}

\subsection{Type I,II error control for one-sided composite hypotheses}

Consider a one-sided testing problem similar to that in Section~\ref{sec:comp_null}, but with a separation between the null hypothesis and alternative
\begin{align}
H_0: \theta^* \leq \theta_0 \quad \text{vs.} \quad H_1:\theta^* \geq \theta_1 \quad (\theta_0<\theta_1). \label{eq:comp_null_sep}
\end{align}
If the monotone likelihood ratio property \eqref{eq:monotone_LR} is satisfied, we can ensure type I and type II error control by using the boosted SPRT of $\theta_1$ vs. $\theta_0$.

\begin{proposition}\label{prop:type_I_II_two_sided}
    Consider the one-sided testing problem with separation \eqref{eq:comp_null_sep} under the monotone likelihood ratio property \eqref{eq:monotone_LR}. Then the boosted SPRT $\delta^{\boost}$ described in Theorem~\ref{theo:boosting_type_I_II} with $\mathbb{P}_1=\mathbb{P}_{\theta_1}$ and $\mathbb{P}_0=\mathbb{P}_{\theta_0}$ provides type I error control for all $\theta\leq \theta_0$ and type II error control for all $\theta\geq \theta_1$.
\end{proposition}

\section{Discussion\label{sec:discussion}}

\subsection{Testing a forecaster\label{sec:forecaster}}

The classical SPRT by Wald \citep{wald1945sequential, Wald1947} assumes that the null hypothesis and the alternative hypothesis are specified in advance in the form of simple distributions $\mathbb{P}_0$ and $\mathbb{P}_1$, respectively.  However, in Section~\ref{sec:comp_alt} we have seen that the SPRT remains even valid if we replace the fixed alternative $\mathbb{P}_1$ with a predictable plugin $\mathbb{P}_1^t$ that depends on the data $X_1,\ldots,X_{t-1}\in \mathcal{X}$ observed so far. In general, the null distribution $\mathbb{P}_0^t$ and the alternative distribution $\mathbb{P}_1^t$ considered at each step $t\in \mathbb{N}$ can be chosen freely based on the previous data $X_1,\ldots, X_{t-1}$ and some additional information $Z_1,\ldots,Z_t\in \mathcal{Z}$.  \citet{shafer2021testing} formulated this as a forecaster-skeptic testing protocol, where the forecaster announces a probability distribution $\mathbb{P}_0^t$ at each step $t$ and a skeptic tests the forecaster by specifying an alternative $\mathbb{P}_1^t$. For example, one can think of the forecaster as a weather forecaster, which is tried to be invalidated by a skeptic. We summarized this general testing protocol in Algorithm~\ref{alg:forecaster} (see supplement). 

The extension of our boosting technique (Algorithm~\ref{alg:general}) to the general testing problem in Algorithm~\ref{alg:forecaster} is straightforward, since at each step $t$ we only require the current martingale value $M_{t-1}=\Lambda_{t-1}$ and the distribution of the next factor $L_t=\frac{p_1^t(X_t)}{p_0^t(X_t)}$ under the null hypothesis $\mathbb{P}_0^t$, where $p_0^t$ and $p_1^t$ are the densities of $\mathbb{P}_0^t$ and $\mathbb{P}_1^t$, respectively.

\subsection{Other approaches to avoid Wald's approximated thresholds\label{sec:other_approachs}}

The problem of \enquote{overshoot} in the SPRT was already noted by \citet{wald1945sequential, Wald1947} and examined in many previous works.  
One line of work approximates the exact type I and type II error of SPRTs based on Brownian motion and renewal theory \citep{siegmund1975error, siegmund2013sequential, woodroofe1976renewal, lai1977nonlinear}. This approach was particularly driven by Siegmund, who gives a comprehensive overview in his book \citep{siegmund2013sequential}. The provided approximations are asymptotically (for $\alpha\to 0$ and $\beta \to 0$) correct under certain assumptions. In Supplementary Material~\ref{sec:siegmund}, we compare our boosted SPRT to the SPRT with Siegmund's approximations in a simple Gaussian setting (indeed, this is the only setting where the application of Siegmund's approximations seems to be straightforward), which shows that the latter method requires even less sample size. Hence, if Siegmund's approximations apply and are reliable, we recommend their use. However, the advantage of our proposed boosting technique is its much broader applicability.  For example, Siegmund's approximations cannot be applied to the one-sided SPRT without separation (Section~\ref{sec:comp_alt}) or any of the applications provided in Section~\ref{sec:extensions}, since the employed test statistics are not independent (if $(\delta_t)_{t\in \mathbb{N}}$ and $(\kappa_t)_{t\in \mathbb{N}}$ are predictable in Sections~\ref{sec:example_CS} and \ref{sec:conformal_martingales}). In contrast, we do not make any separation or independence assumptions. 
Also note that in contrast to Siegmund's approximations our boosting technique provides non-asymptotic validity.

Another approach is to truncate the SPRT by prespecifying a maximum number of observations $N$. For simple distributions of the likelihood ratio and sufficiently small $N$ it is then possible to calculate the type I and type II error for given thresholds exactly \citep{armitage1969repeated}, which can be used to find appropriate bounds for the SPRT. Such an approach is frequently employed in group sequential trials \citep{jennison1999group, wassmer2016group}. For more complex distributions one could approximate the error probabilities using simulations. This often leads to considerable computational effort, which sometimes can be reduced using importance sampling \citep{siegmund1976importance}.

One limitation all the above approaches have in common is that the test statistics and their distributions must be specified in advance. This is not possible in the general forecaster-skeptic testing protocol described in the previous section. Hence, while our boosting approach is applicable to such complex and adaptive testing protocols, all of the aforementioned methods fail. Furthermore, after calculating a boosting factor $b_t$ with our method, one can always plug it back into the expectation (\eqref{eq:boosting_inequality} in the one-sided or \eqref{eq:boost_original} and \eqref{eq:boost_inverse} in the two-sided case) and only accept the boosting factor if the expectation is indeed bounded by $1$. If a boosting factor cannot be accepted, one can try to calculate another boosting factor or set $b_t=1$ if no greater boosting factor can be found. In this way, one can always ensure that the resulting test is conservative, meaning that it provides the desired type I and type II error. This is neither possible with Siegmund's approach nor when using simulations.  

Moreover, our approach provides a test supermartingale, which yields an e-value at stopping times \citep{ramdas2023game, grunwald2020safe, shafer2021testing}. For example, this can be important for optional continuation \citep{grunwald2020safe, shafer2021testing}, combining the evidence with other test results \citep{vovk2021values} and in multiple testing \citep{xu2025bringing, goeman2025partitioning}.

\section{Summary}

We introduced a general approach to uniformly improve nonnegative (super)martingale based sequential tests by avoiding an overshoot at the upper threshold $1/\alpha$ (where $\alpha$ is the desired type I error). We illustrated its application by uniformly improving the conservative power-one SPRT for one-sided null and alternative hypotheses. We also obtained an alternative to the two-sided SPRT with Wald's approximated thresholds, where our alternative provides valid type I and type II error control while often reducing the sample size. Our sequential tests can also be used to tighten confidence sequences and apply in nontrivial settings like sampling without replacement and conformal martingales.

 
 
 



\subsection*{Acknowledgments}
 LF acknowledges funding by the Deutsche Forschungsgemeinschaft (DFG, German Research Foundation) – Project number 281474342/GRK2224/2. AR was funded by NSF grant DMS-2310718. We are grateful to Wouter M. Koolen for pointing out the possibility of using an inverted SPRT to control the type II error, and to Nick Koning for mentioning the log-optimality of boosting factors. In addition, we would like to thank Yajun Mei, Qunzhi Xu, Venugopal Veeravalli and Georgios Fellouris for useful discussions.

\bibliography{main}
\bibliographystyle{plainnat}

\clearpage

\begin{center}
    {\LARGE \bfseries Supplementary material for ``Improving Wald's (approximate) sequential probability ratio test by avoiding overshoot''} \\[1ex] %
\end{center}

\renewcommand{\thesection}{S.\arabic{section}}
\renewcommand{\theequation}{S.\arabic{equation}}
\renewcommand{\thetheorem}{S.\arabic{theorem}}
\renewcommand{\theproposition}{S.\arabic{proposition}}
\renewcommand{\thecorol}{S.\arabic{corol}}
\renewcommand{\thelemma}{S.\arabic{lemma}}
\renewcommand{\theremark}{S.\arabic{remark}}
\renewcommand{\theexample}{S.\arabic{example}}
\renewcommand{\thefigure}{S.\arabic{figure}}
\renewcommand{\thetable}{S.\arabic{table}}
\renewcommand{\thealgorithm}{S.\arabic{algorithm}}

\setcounter{section}{0}
\setcounter{equation}{0}
\setcounter{theorem}{0}
\setcounter{remark}{0}
\setcounter{example}{0}
\setcounter{figure}{0}
\setcounter{table}{0}
\setcounter{algorithm}{0}

\section{Boosting a confidence sequence for Gaussian means\label{sec:example_CS}}

    Suppose we desire a CS for the mean $\mu^*$ of i.i.d. Gaussians $X_1,X_2,\ldots$ with unit variance. For each $\mu\in \mathbb{R}$, we need to determine a sequential test for
 $
 H_0^{\mu}:\mu^*=\mu.  
 $
 Defining $Z_t=X_t-\mu$, note that each $Z_t$ is standard normally distributed under $H_0^{\mu}$. Thus, 
$
\mathbb{E}_{\mu}[\exp(\delta Z_t - \delta^2/2)|\mathcal{F}_{t-1}]= 1,
$
where $\delta>0$ is any constant. Hence, 
$$
M_t^{\mu}=\prod_{i=1}^t \exp(\delta Z_i - \delta^2/2)
$$
is a test martingale for $H_0^{\mu}$. Further, $M_t^{\mu}$ is decreasing in $\mu$ for each $t$, so the family of test martingales $((M_t^{\mu}))_{\mu\in \mathbb{R}}$ can be used to find a lower bound $l_t$ for $\mu^*$. Indeed, we can calculate
\begin{align*}
   M_t^{\mu}< 1/\alpha  \Leftrightarrow \prod_{i=1}^t \exp(\delta (X_i-\mu) - \delta^2/2) < 1/\alpha \Leftrightarrow \frac{1}{t} \sum_{i=1}^t X_i-\frac{\log(1/\alpha)}{n\delta} -\delta/2 < \mu.
\end{align*}
Thus, $(l_t)_{t\in \mathbb{N}}$, where $l_t=\frac{1}{t} \sum_{i=1}^t X_i-\frac{\log(1/\alpha)}{n\delta} -\delta/2$, is a valid lower confidence sequence for $\mu^*$. Setting $Z_t=\mu-X_t$ gives an upper bound, yielding the $(1-\alpha)$ confidence sequence
\begin{align}
\frac{1}{t} \sum_{i=1}^t X_i\ \pm \ \left[\frac{\log(2/\alpha)}{t \delta} + \delta/2\right]. \label{eq:howard}
\end{align}
This CS was already proposed by Robbins and colleagues \citep{darling1967confidence, robbins1970statistical}.
One can then obtain a boosted confidence sequence by boosting each $(M_t^{\mu})$, $\mu\in \mathbb{R}$, separately and choosing the smallest/largest $\mu\in \mathbb{R}$ such that $M_t^{\mu,\boost}< 1/\alpha$ for the lower/upper bound. In Figure~\ref{fig:sim_CS} we compare the $(1-\alpha)$ lower confidence bound proposed by Robbins with our boosted improvement for $\mu^*=2$ (the bounds were obtained by averaging over $100$ trials), which shows that boosting can tighten the CS \eqref{eq:howard} significantly. We chose $\delta=\sqrt{8\log(1/\alpha)/n}$ for $n=50$, as recommended for obtaining a tight CS at time $n$ \citep{howard2020time}.

 \begin{remark}
     We just assumed a constant $\delta>0$ for simplicity. In general, it is possible to choose a predictable sequence $(\delta_t)_{t\in \mathbb{N}}$ and set $M_t^{\mu}=\prod_{i=1}^t \exp(\delta_i Z_i-\delta_i^2/2)$. For specific  $(\delta_t)_{t\in \mathbb{N}}$, this yields CSs with the desirable property that their width $u_t-l_t$ shrinks to $0$ \citep{waudby2023estimating}. Our boosting approach can be applied to such predictable plugin CSs in the exact same manner as above (see also  Section~\ref{sec:comp_alt}). 
 \end{remark}

\begin{remark}
     Note that Robbins proved that the CS \eqref{eq:howard} is valid for all 1-sub-Gaussian distributions. However, our boosting approach only ensures valid error control for Gaussian distributions with variance one. The reason is that $$\mathbb{E}_0[L_t'|\mathcal{F}_{t-1}] \leq \mathbb{E}_0[L_t|\mathcal{F}_{t-1}] \centernot\implies \mathbb{E}_0[T_t(b_tL_t';M_{t-1}^{\boost})|\mathcal{F}_{t-1}] \leq \mathbb{E}_0[T_t(b_tL_t; M_{t-1}^{\boost})|\mathcal{F}_{t-1}]$$
for $b_t> 1$ and therefore it is not possible to handle all 1-sub-Gaussian distributions simultaneously.  
However, we think it is interesting that it is possible to exploit the information of the observations being exactly Gaussian (instead of just 1-sub-Gaussian), since previous works on Gaussian or sub-Gaussian CSs did not make use of it \citep{darling1967confidence, robbins1970statistical, howard2021time}.
\end{remark}

\begin{figure}[h!]
\centering
\includegraphics[width=0.7\textwidth]{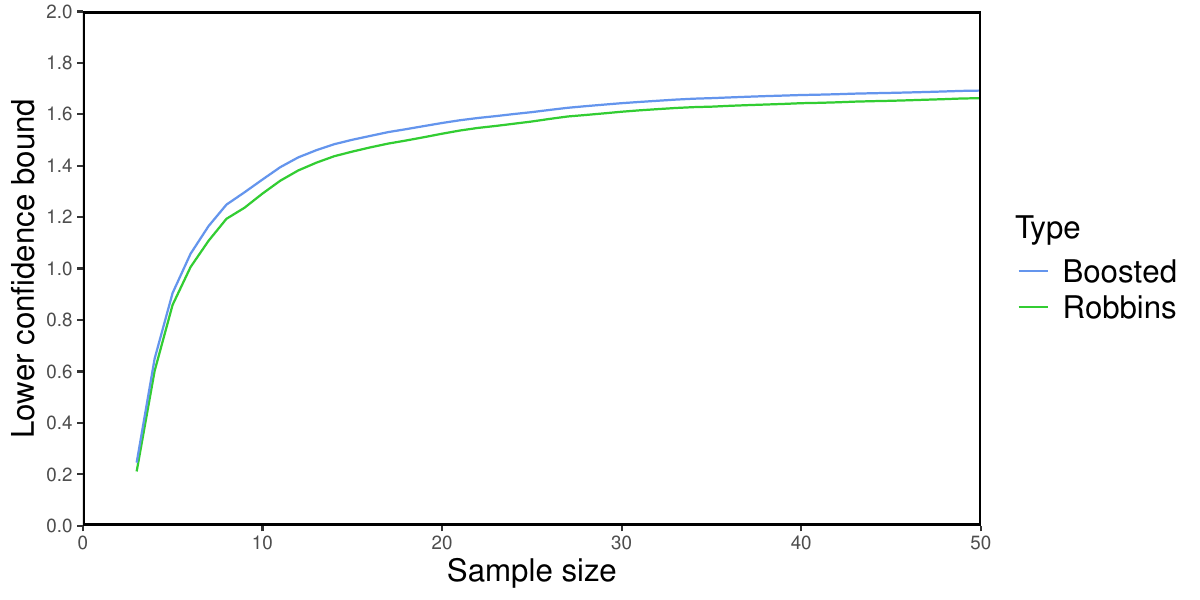}
\caption{Comparison of the lower confidence bound \eqref{eq:howard} by \citet{howard2020time} with its boosted improvement for $\mu^*=2$. Boosting improves the CS significantly. \label{fig:sim_CS} }\end{figure}

\section{Additional simulation results\label{appn:sims}}
In this section, we provide some additional simulation result. In Figure~\ref{fig:sim_simple_small_alpha} we repeated the simulated experiment for simple hypotheses from Figure~\ref{fig:sim_simple} but with $\alpha=0.01$ instead of $\alpha=0.05$. Obviously, the required sample size is larger for $\alpha=0.01$ than for $\alpha=0.05$. It can also be seen that the gain in percentage of sample size is smaller than for $\alpha=0.05$, however, due to the larger total sample size the absolute gain in sample size remains similar. 

\begin{figure}[h!]
\centering
\includegraphics[width=\textwidth]{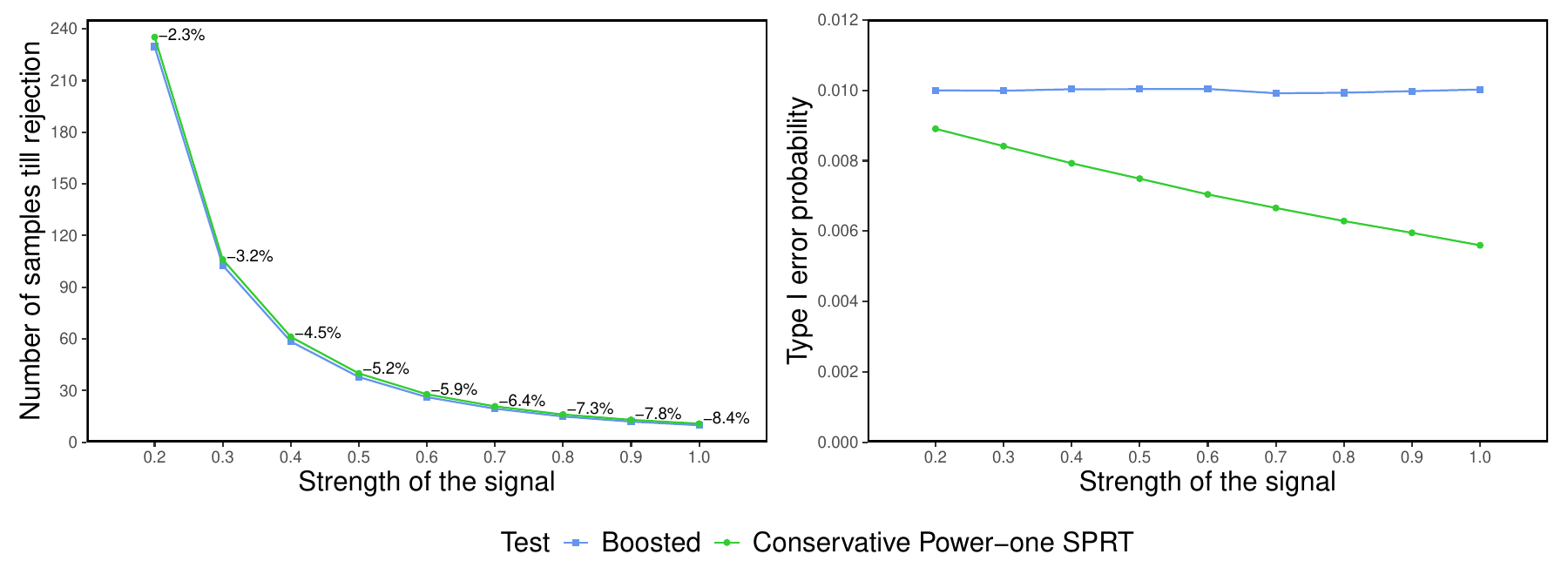}
\caption{Comparison of sample size and type I error between the SPRT and the boosted SPRT in a simple Gaussian testing setup for different signal strengths and $\alpha=0.01$. The boosted SPRT saves $2.3-8.4\%$ of sample size compared to the SPRT and exhausts the significance level. The SPRT is becoming increasingly conservative for stronger signals.  \label{fig:sim_simple_small_alpha} }\end{figure}

In Figure~\ref{fig:sim_futility_small_alpha}, we repeated the two-sided experiment from Figure~\ref{fig:sim_futility} for $\alpha=0.01$. Again, the boosted SPRT saves a substantial number of samples compared to Wald's two-sided SPRT with approximated thresholds and nearly exhausts the type I and type II error. Interestingly, the percentage of saved samples approximately increases with $\beta$ for $\alpha=0.01$, while it decreases for $\alpha=0.05$. 

\begin{figure}[h!]
\centering
\includegraphics[width=\textwidth]{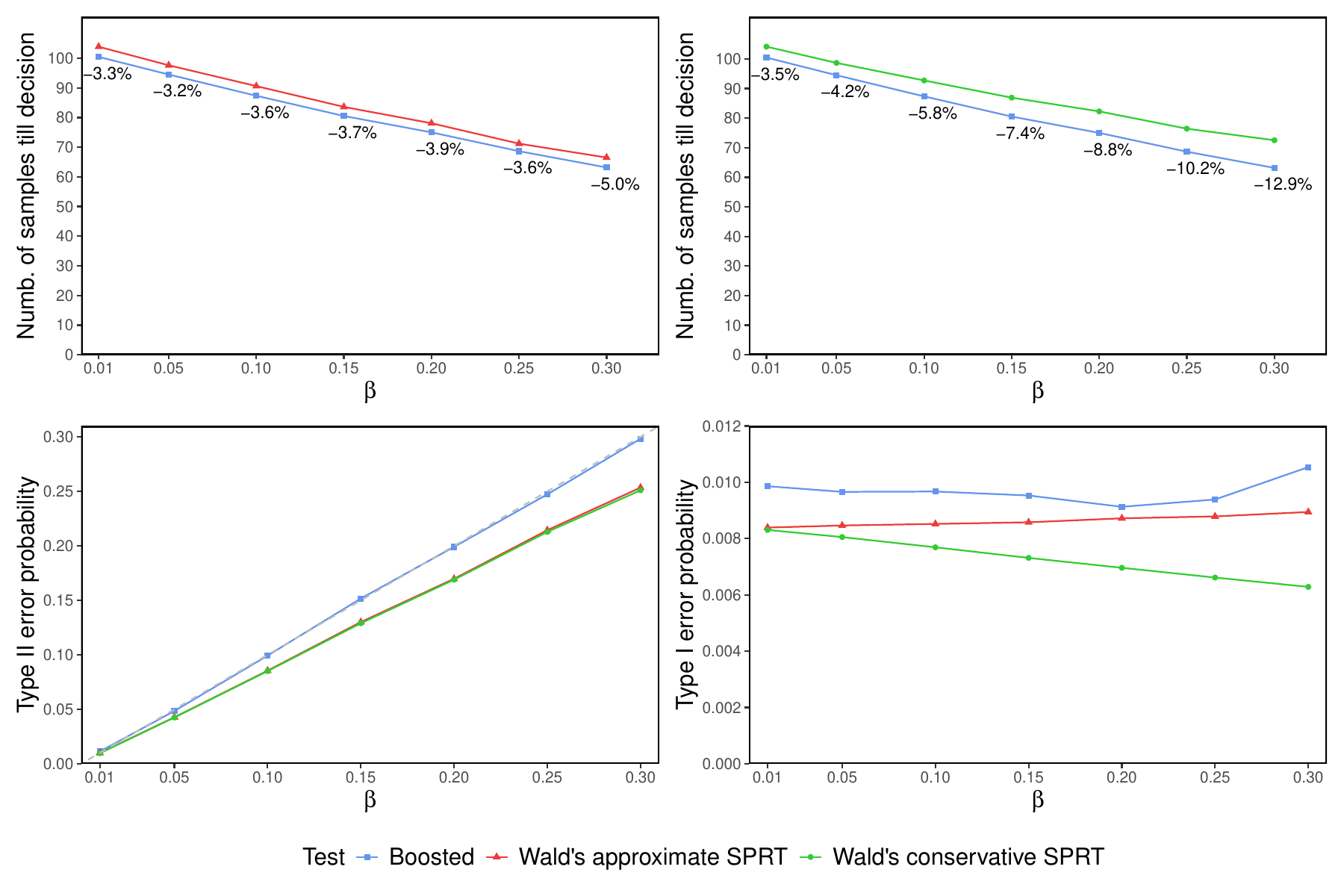}
\caption{Comparison of sample size, type II and type I error probability between the Wald's approximated SPRT, Wald's conservative SPRT and the boosted SPRT in a simple Gaussian testing setup for different desired type II error probabilities ($\beta$)  and $\alpha=0.01$. Wald's conservative SPRT and our boosted method provably control the type I and type II error, while the SPRT with Wald's approximated thresholds does not. Furthermore, the boosted SPRT saves sample size and nearly exhausts the error probabilities. \label{fig:sim_futility_small_alpha} }\end{figure}

\section{Boosting with a stop for futility and discrete data\label{sec:boosting_discrete}}

In Section~\ref{sec:two_sided_typeI}, we have shown that with discrete data and two-sided thresholds it might not be possible to find boosting factors such that the expected value of the boosted (super)martingale factors is exactly one. In this section, we propose a randomization approach that can be used to avoid conservative type I error control in these cases.

First, at each step $t$ choose the largest $b_t^*\geq 1$ among all $b_t\geq 1$ that ensures 
$$\Ex_{0}[T_{\alpha}(b_t^* L_t;M_{t-1}^{\mathrm{boost}}, \nu_t)|\mathcal{F}_{t-1}]\leq 1.$$ 
Note that such a $b_t^*$ always exists, since $T_{\alpha}(x;M_{t-1}^{\mathrm{boost}}, \nu_t)$ is left-continuous and nondecreasing in $x$. Then, in case of $b_t^* L_t M_{t-1}^{\mathrm{boost}}\leq \nu_t$, draw a random sample $U\sim [0,1]$ and reject $H_0$ if $U\leq \alpha a_t M_{t-1}^{\mathrm{boost}}$, where
\begin{align*}
a_t= \frac{1-\Ex_{0}[T_{\alpha}(b_t^* L_t;M_{t-1}^{\mathrm{boost}}, \nu_t)|\mathcal{F}_{t-1}]}{\mathbb{P}_0(b_t^*L_tM_{t-1}^{\mathrm{boost}}\leq \nu_t|\mathcal{F}_{t-1}) }. 
\end{align*}
Hence, even if our boosted process $(M_t^{\boost})$ crosses the lower threshold $\nu_t$, with this approach we still have a positive probability for rejection if $\Ex_{0}[T_{\alpha}(b_t^* L_t;M_{t-1}^{\mathrm{boost}}, \nu_t)|\mathcal{F}_{t-1}]$ is strictly smaller than one. 

To see why this controls the type I error, we need the randomized Ville's inequality \citep{ramdas2023randomized}, a randomized improvement of Ville's inequality \eqref{eq:ville}. It states that for any test supermartingale $(M_t)$ with respect to the filtration $(\mathcal{F}_t)$, we have that
$\mathbb{P}_0(M_{\tau}\geq U/\alpha)\leq \alpha$ for all stopping times, where $U$ is uniform on $[0,1]$ and independent of $(\mathcal{F}_t)$. Now define
\begin{align*}
 \tilde{T}(b_t^* L_t;M_{t-1}^{\mathrm{boost}}, \nu_t)\coloneqq a_t\mathbbm{1}\{b_t^*L_tM_{t-1}^{\mathrm{boost}}\leq \nu_t\}+T_{\alpha}(b_t^* L_t;M_{t-1}^{\mathrm{boost}}, \nu_t).
\end{align*}
Then $\Ex_{0}[\tilde{T}_{\alpha}(b_t^* L_t;M_{t-1}^{\mathrm{boost}}, \nu_t)|\mathcal{F}_{t-1}]= 1$. Now if we use $\tilde{T}_{\alpha}(b_t^* L_t;M_{t-1}^{\mathrm{boost}}, \nu_t)$ instead of $T_{\alpha}(b_t^* L_t;M_{t-1}^{\mathrm{boost}}, \nu_t)$ as factor for our boosted test martingale, the test martingale has the same value if we did not stop for futility but it equals $a_t M_{t-1}^{\mathrm{boost}}$ (instead of $0$) in case of $b_t^*L_tM_{t-1}^{\mathrm{boost}}\leq \nu_t$. Hence, due to randomized Ville's inequality \citep{ramdas2023randomized}, after stopping in case of $b_t^*L_tM_{t-1}^{\mathrm{boost}}\leq \nu_t$, we can reject the null hypothesis if $U\leq \alpha a_t M_{t-1}^{\mathrm{boost}}$ while ensuring valid type I error control.

\begin{example}[Continuation of Example~\ref{example:binary}]
    Suppose we have i.i.d. binary data $X_1,X_2,\ldots$ with $\mathbb{E}_0[X_t]=1/3$, $\mathbb{E}_1[X_t]=2/3$, $\alpha=0.05$, $\Lambda_{t-1}^{\boost}=8$ and $\nu_t=7$. Then the largest $b_t\geq 1$ satisfying $\Ex_{0}[T_{\alpha}(b_t L_t;\Lambda_{t-1}^{\mathrm{boost}}, \nu_t)|\mathcal{F}_{t-1}]\leq 1$ is given by $b_t^*=7/4$. With this, we have $\Ex_{0}[T_{\alpha}(b_t^* L_t;\Lambda_{t-1}^{\mathrm{boost}}, \nu_t)|\mathcal{F}_{t-1}]= 5/6$ and $\mathbb{P}_0(b_t^* L_t\Lambda_{t-1}^{\mathrm{boost}}\leq \nu_t|\mathcal{F}_{t-1})=2/3$. Hence, $a_t=(1-5/6)/(2/3)=1/4$. Therefore, we can reject $H_0$ in case of $X_t=0$ if $U\leq \alpha a_t \Lambda_{t-1}^{\mathrm{boost}}= 1/10$ (in case of $X_t=1$ we have $\Lambda_{t}^{\mathrm{boost}}=20$ and can reject anyway). 
\end{example}

\section{Implementation of the boosted two-sided SPRT\label{sec:implementation_2side}}

In this section we provide some details on our implementation of the two-sided boosting method proposed in Theorem~\ref{theo:boosting_type_I_II} that we, for example, used to generate Figure~\ref{fig:sim_futility}. 

According to Proposition~\ref{prop:boosting_equal_2side}, one possibility would be to just replace the inequalities in \eqref{eq:boost_original} and \eqref{eq:boost_inverse} with equalities and then use a solver for nonlinear equation systems to find the optimal boosting factors. Appropriate solvers for this purpose are provided by the \texttt{R} packages \textit{nleqslv} \citep{nleqslv} and \textit{BB} \citep{BB}. While these work most of the time, we noticed that the algorithms sometimes fail to find a solution and then produce a warning, e.g. if our current boosted likelihood ratio value $\Lambda_{t-1}^{\boost}$ is very close to $1/\alpha$. For this reason, in our implementation we used a slightly different approach by formulating the following optimization problem with nonlinear constraints
$$
\max_{b_t\geq 1, b_t^{\mathrm{inv}}\geq 1} b_t+b_t^{\textrm{inv}} \quad \text{ s.t. \eqref{eq:boost_original} and \eqref{eq:boost_inverse} are satisfied}.
$$

We solved this using the \textit{NLOPT\_LN\_COBYLA} algorithm from the \textit{nloptr} \citep{nloptr} \texttt{R} package. While it produced the same results as the previously mentioned approaches in most cases, we found it to be more reliable in extreme scenarios (e.g. when our current boosted likelihood ratio value is very close to one of the thresholds). To avoid inflated error probabilities by all means, we verify every boosting factor by plugging it into \eqref{eq:boost_original} and \eqref{eq:boost_inverse}. If one of the two expected values is greater than $1+10^{-7}$, we set both boosting factors to $1$. This ensures that we always have (sufficiently accurate) type I and type II error control, even if the optimization algorithm fails to provide valid solutions. In order to avoid conservatism by setting the boosting factors to one in those cases, one could also consider to apply the different algorithms one after the other if one fails to provide a sufficient solution. However, in our simulations we found the described proceeding to perform sufficiently well.

\section{General forecaster-skeptic testing protocol}

\begin{algorithm}
\caption{SPRT in a general forecaster-skeptic testing protocol}\label{alg:forecaster}
\hspace*{\algorithmicindent} \textbf{Input:} Lower and upper SPRT thresholds $\gamma_0$ and $\gamma_1$.\\
 \hspace*{\algorithmicindent} 
 \textbf{Output:} Decision on the forecaster.
\begin{algorithmic}[1]
\State $\Lambda_0=1$
\For{$t=1,2,\ldots$} 
\State Reality announces $Z_t\in \mathcal{Z}$
\State Forecaster announces a probability distribution $\mathbb{P}_0^t$ on $\mathcal{X}$
\State Skeptic announces a probability distribution $\mathbb{P}_1^t$ on $\mathcal{X}$
\State Reality announces $X_t\in \mathcal{X}$
\State $\Lambda_t=\Lambda_{t-1} \cdot \frac{p_1^t(X_t)}{p_0^t(X_t)}$
\If{$\Lambda_t\geq \gamma_1$}
\State Stop and reject the forecaster
\EndIf
\If{$\Lambda_t\leq \gamma_0$}
\State Stop and accept the forecaster
\EndIf
\EndFor
\end{algorithmic}
\end{algorithm}

\section{Comparison to Siegmund's SPRT\label{sec:siegmund}}

In this section, we compare our boosted SPRT to the SPRT with the approximated thresholds by \citet{siegmund2013sequential}. We consider a simple i.i.d. Gaussian testing setup $H_0:X_t\sim \mathcal{N}(0,1)$ vs. $H_1:X_t\sim \mathcal{N}(\mu_1,1)$, where $\mu_1>0$ is the strength of the signal. Siegmund's approximations suggest to choose $\gamma_1=\frac{1-\beta}{\alpha \exp(\mu_1 \rho)}$ and $\gamma_0=\frac{\beta\exp(\mu_1 \rho)}{1-\alpha}$, where $\rho=0.583$. In Figure~\ref{fig:sim_simple_siegmund}, we compare the SPRT with these approximated thresholds to our boosted SPRT in the power-one ($\beta=0$) case and in Figure~\ref{fig:sim_futility_siegmund} for different desired type II errors $\beta$. In both cases $\alpha=0.05$ and the implementation of our boosted SPRT is exactly the same as in Figures~\ref{fig:sim_simple} and \ref{fig:sim_futility}.

It is easy to see that Siegmund's SPRT outperformed our boosted SPRT in all considered cases. While the difference in the power-one scenario is hardly visible, the gain in sample size for larger $\beta$ can be substantial. Note that the type II error in Figure~\ref{fig:sim_futility_siegmund} seems slightly inflated for Siegmund's SPRT, however, this should be negligible for practice. Hence, in this particular case Siegmund's SPRT may be preferred over our boosted SPRT. However, note that Siegmund's approximations are only asymptotically valid (for $\alpha\to 0$ and $\beta\to 0$) and hence it might be sensible to check their validity by such simulations before application, particularly if $\beta$ and/or $\alpha$ is large. Furthermore, the above thresholds, in particular the value of $\rho=0.583$, only work in the Gaussian example and require (nontrivial) adjustments for other distributions. In general, it is not possible at all to use Siegmund's approximations in settings with non-i.i.d. test statistics such as considered in Sections~\ref{sec:comp_alt} and \ref{sec:WoR}.

\begin{figure}[h!]
\centering
\includegraphics[width=\textwidth]{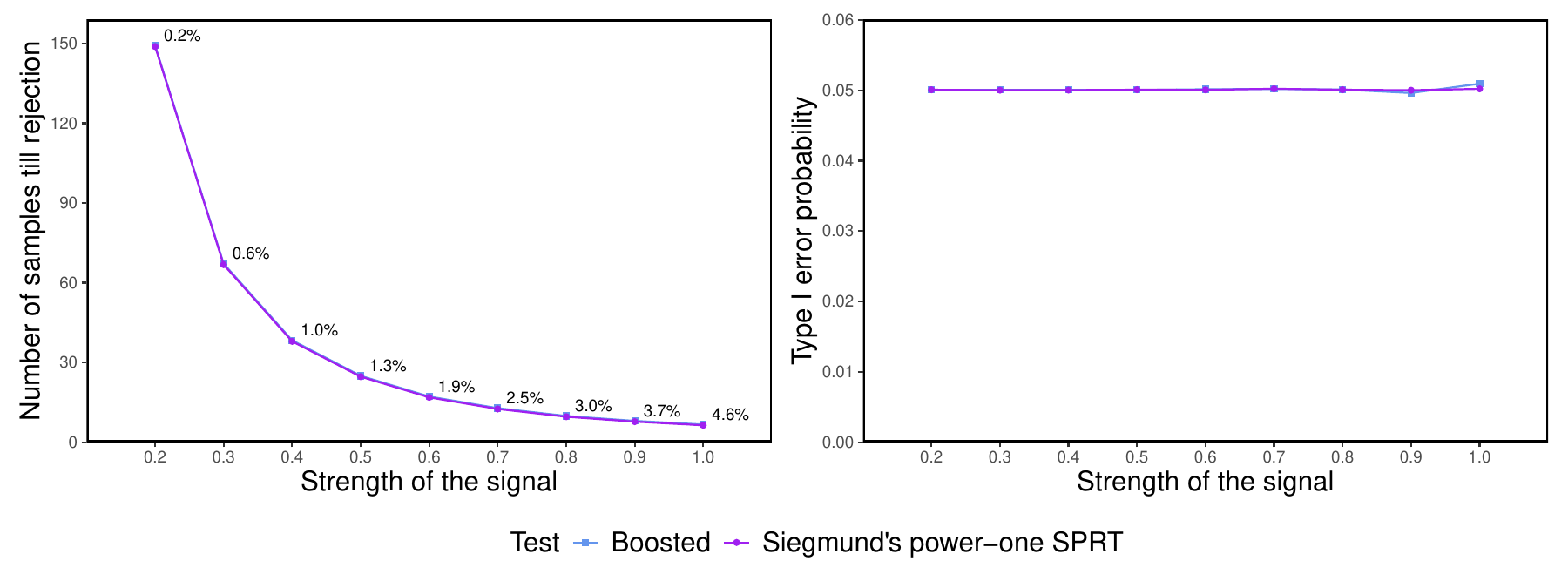}
\caption{Comparison of sample size and type I error between Siegmund's SPRT and the boosted SPRT in a simple Gaussian testing setup for different signal strengths and $\alpha=0.05$. Siegmund's SPRT saves sample size compared to the boosted SPRT, but the gain is hardly visible. Siegmund's SPRT controls the type I error asymptotically, while the boosted SPRT provides non-asymptotic control.  \label{fig:sim_simple_siegmund} }\end{figure}

\begin{figure}[h!]
\centering
\includegraphics[width=\textwidth]{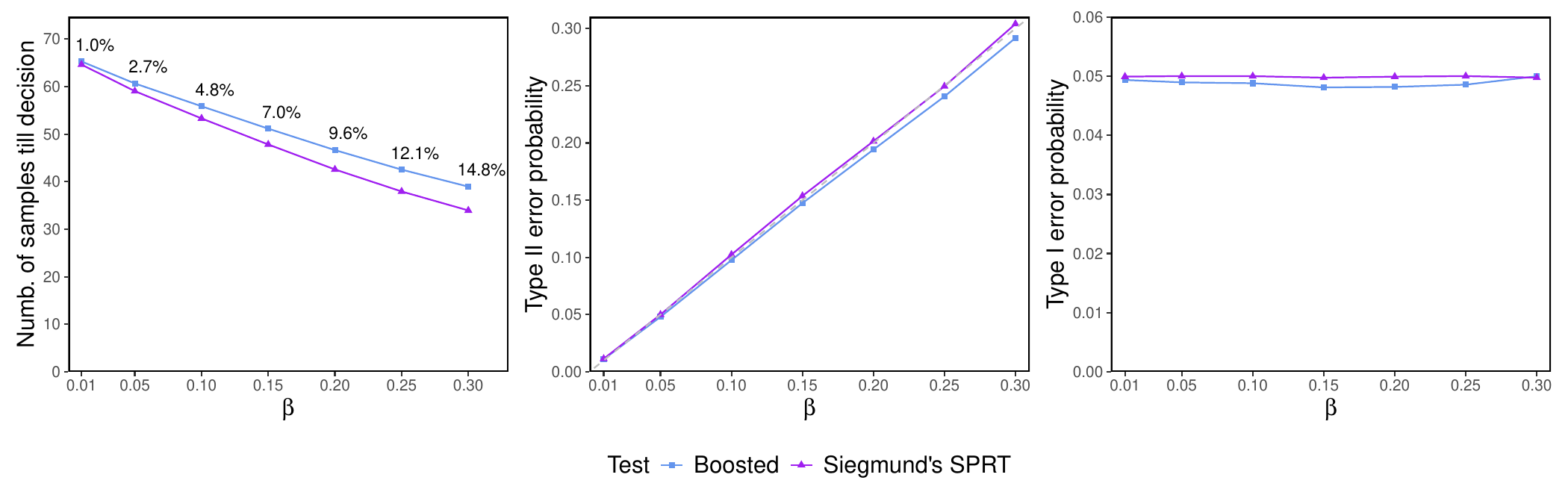}
\caption{Comparison of sample size, type II and type I error probability between the SPRT with Siegmund's approximated thresholds and the boosted SPRT in a simple Gaussian testing setup for different desired type II error probabilities ($\beta$) and $\alpha=0.05$. Siegmund's SPRT needs substantially less samples than the boosted SPRT. Siegmund's SPRT controls the error probabilities asymptotically, while the boosted SPRT provides non-asymptotic control. \label{fig:sim_futility_siegmund} }\end{figure}

\section{Proofs}
\begin{proof}[Proof of Theorem~\ref{theo:main}]
    We already noted that $\tau_{\mathrm{trunc}}=\tau_M$ 
 and $M_t^{\mathrm{trunc}}=M_t$ for all $t<\tau_M$. Since $M_t^{\mathrm{boost}}\geq M_t^{\mathrm{trunc}}$ for all $t\in \mathbb{N}$, the first claim follows. Furthermore, since $\Ex_{0}[T_{\alpha}(b_t L_t;M_{t-1}^{\mathrm{boost}})|\mathcal{F}_{t-1}]\leq 1$ for all $t\in \mathbb{N}$, it immediately follows that $(M_t^{\boost})$ is a test supermartingale.

\end{proof}

\begin{proof}[Proof of Proposition~\ref{prop:solvability}]
   First, note that $T_{\alpha}(x ; M_{t-1}^{\boost})$ is a continuous function in $x$ and upper bounded by $1/(\alpha M_{t-1}^{\boost})$. Hence, by the dominated convergence theorem we have for any $b_0\geq 1$,
   $$
   \lim\limits_{b\to b_0} \Ex_{0}[T_{\alpha}(b L_t;M_{t-1}^{\mathrm{boost}})|\mathcal{F}_{t-1}] = \Ex_{0}\left[ \lim\limits_{b\to b_0}T_{\alpha}(b L_t;M_{t-1}^{\mathrm{boost}})\bigg\vert\mathcal{F}_{t-1}\right] =\Ex_{0}\left[ T_{\alpha}(b_0 L_t;M_{t-1}^{\mathrm{boost}})|\mathcal{F}_{t-1}\right],
   $$
   implying that $b\mapsto \Ex_{0}[T_{\alpha}(b L_t;M_{t-1}^{\mathrm{boost}})|\mathcal{F}_{t-1}]$ is a continuous function as well. Since $b_t=1$ always solves \eqref{eq:boosting_inequality} and $\lim\limits_{b_t\to \infty} \mathbb{E}_{0}[T_{\alpha}(b_tL_t ; M_{t-1}^{\boost})| \mathcal{F}_{t-1}]=\mathbb{P}_0(L_t>0|\mathcal{F}_{t-1})/(M_{t-1}^{\boost} \alpha)>1$, there exists a $b_t^*$ that satisfies \eqref{eq:boosting_equality}. Now let $b_t \neq b_t^*$ be another solution to \eqref{eq:boosting_inequality}. In case of $b_t<b_t^*$, it immediately follows that $T_{\alpha}(b_t^* L_t;M_{t-1}^{\mathrm{boost}})\geq T_{\alpha}(b_t L_t;M_{t-1}^{\mathrm{boost}})$ $\mathbb{P}_0$-almost surely. If $b_t>b_t^*$, then $T_{\alpha}(b_t L_t;M_{t-1}^{\mathrm{boost}})\geq T_{\alpha}(b_t^* L_t;M_{t-1}^{\mathrm{boost}})$ $\mathbb{P}_0$-almost surely and $\Ex_{0}[T_{\alpha}(b_t L_t;M_{t-1}^{\mathrm{boost}})|\mathcal{F}_{t-1}]=1=\Ex_{0}[T_{\alpha}(b_t^* L_t;M_{t-1}^{\mathrm{boost}})|\mathcal{F}_{t-1}]$, implying that $T_{\alpha}(b_t L_t;M_{t-1}^{\mathrm{boost}})= T_{\alpha}(b_t^* L_t;M_{t-1}^{\mathrm{boost}})$ $\mathbb{P}_0$-almost surely. 
   
\end{proof}

\begin{proof}[Proof of Proposition~\ref{prop:one_sided_null}]
Due to the monotone likelihood ratio property~\eqref{eq:monotone_LR}, we know that 
\begin{align}
\mathbb{P}_{\theta_0}\left(\lambda_t(X_t)\geq x\right)\geq \mathbb{P}_{\theta}\left(\lambda_t(X_t)\geq x\right) \quad (\theta \leq \theta_0, x\geq 0), \label{eq:Karlin_rubin}
\end{align}
where $\theta_t\geq \theta_0$. This implies that
$$
\mathbb{P}_{\theta_0}\left(T_{\alpha}(b_t\lambda_t(X_t);M)\geq x\right)\geq \mathbb{P}_{\theta}\left(T_{\alpha}(b_t\lambda_t(X_t);M)\geq x\right) \quad (\theta \leq \theta_0, x\geq 0)
$$
for any $\alpha\in (0,1)$, $M>0$ and $b_t\geq 1$. Hence, if 
$
\Ex_{\theta_0}[T_{\alpha}(b_t\lambda_t(X_t);M)|\Lambda_{t-1}^{\boost}]\leq 1
$, we also have $\Ex_{\theta}[T_{\alpha}(b_t\lambda_t(X_t);M)|\Lambda_{t-1}^{\boost}]\leq 1$ for all $\theta<\theta_0$. 
\end{proof}

\begin{proof}[Proof of Theorem~\ref{theo:main_2side}]
    First, note that $M_t^{\boost}\leq \nu_t$ implies that $\prod_{i=1}^{t'} b_i L_i\leq \nu_t^M \prod_{i=1}^{t'} b_i$ for some $t'\leq t$, and hence $M_{t'}\leq \nu_{t'}^M$. Therefore, if $t<\tau_M$, the boosted martingale has not stopped for futility yet, which shows that  $M_t^{\boost}\geq M_t$ for all $t<\tau_M$ and $\delta_{\boost}\geq \delta$ in the same manner as in Theorem~\ref{theo:main}. We now show that $M_t\leq \nu_t^M$ implies that $\tau_{\boost}\leq t$. If $\nu_t=\nu_t^M\prod_{i=1}^t b_i$, then $M_t\leq \nu_t^M$ is equivalent to $M_t \prod_{i=1} b_i \leq \nu_t^M \prod_{i=1}^t b_i=\nu_t$. If $\nu_t=1/\alpha$, then $\tau_{\boost}\leq t$, since we either have $M_t^{\boost}=0$ or $M_t^{\boost}=1/\alpha$. Obviously, $M_t\geq 1/\alpha$ implies that $\tau_{\boost}\leq t$, which shows that $\tau_{\boost}\leq \tau_M$.
    Finally, \eqref{eq:boosting_inequality_2side} immediately implies that $(M_t^{\boost})$ is a test supermartingale. 
\end{proof}

\begin{proof}[Proof of Proposition~\ref{prop:solvability_2side}]

We only show that $b\mapsto \Ex_{0}[T_{\alpha}(b L_t;M_{t-1}^{\mathrm{boost}}, \nu_t)|\mathcal{F}_{t-1}]$ is a continuous mapping. The remaining proof follows in the exact same manner as in Proposition~\ref{prop:solvability}. Let $b_1\geq 1$ and $\epsilon>0$. We want to show that for $b_2>b_1$ sufficiently close to $b_1$ we have that 
$$
\Ex_{0}[T_{\alpha}(b_2 L_t;M_{t-1}^{\mathrm{boost}}, \nu_t)|\mathcal{F}_{t-1}]-\Ex_{0}[T_{\alpha}(b_1 L_t;M_{t-1}^{\mathrm{boost}}, \nu_t)|\mathcal{F}_{t-1}]\leq \epsilon.
$$
For this, we first note that $$\min(b_2x,1/[M_{t-1}^{\mathrm{boost}} \alpha])-\min(b_1x,1/[M_{t-1}^{\mathrm{boost}} \alpha]\leq \epsilon/2 \ \text{for all } x\geq 0$$ for $b_2$ sufficiently close to $b_1$. Let $f_{L_t|\mathcal{F}_{t-1}}$ be the density of $L_t|\mathcal{F}_{t-1}$ under $\mathbb{P}_0$, then
\begin{align*}
   & \Ex_{0}[T_{\alpha}(b_2 L_t;M_{t-1}^{\mathrm{boost}}, \nu_t)|\mathcal{F}_{t-1}]-\Ex_{0}[T_{\alpha}(b_1 L_t;M_{t-1}^{\mathrm{boost}}, \nu_t)|\mathcal{F}_{t-1}] \\
   &= \int\limits_{\frac{\nu_t}{b_2 M_{t-1}^{\mathrm{boost}}}}^\infty \min(b_2x,1/[M_{t-1}^{\mathrm{boost}} \alpha]) f_{L_t|\mathcal{F}_{t-1}}(x) \, dx - \int\limits_{\frac{\nu_t}{b_1 M_{t-1}^{\mathrm{boost}}}}^\infty \min(b_1x,1/[M_{t-1}^{\mathrm{boost}} \alpha]) f_{L_t|\mathcal{F}_{t-1}}(x) \, dx\\
   &\leq \int\limits_{\frac{\nu_t}{b_2 M_{t-1}^{\mathrm{boost}}}}^\infty \min(b_2x,1/[M_{t-1}^{\mathrm{boost}} \alpha]) f_{L_t|\mathcal{F}_{t-1}}(x) \, dx - \int\limits_{\frac{\nu_t}{b_1 M_{t-1}^{\mathrm{boost}}}}^\infty \min(b_2x,1/[M_{t-1}^{\mathrm{boost}} \alpha]) f_{L_t|\mathcal{F}_{t-1}}(x) \, dx + \epsilon/2\\
   &= \int\limits_{\frac{\nu_t}{b_2 M_{t-1}^{\mathrm{boost}}}}^{\frac{\nu_t}{b_1 M_{t-1}^{\mathrm{boost}}}} \min(b_2x,1/[M_{t-1}^{\mathrm{boost}} \alpha]) f_{L_t|\mathcal{F}_{t-1}}(x) \, dx + \epsilon/2 \\
   &\leq 1/[M_{t-1}^{\mathrm{boost}} \alpha] \int\limits_{\frac{\nu_t}{b_2 M_{t-1}^{\mathrm{boost}}}}^{\frac{\nu_t}{b_1 M_{t-1}^{\mathrm{boost}}}} f_{L_t|\mathcal{F}_{t-1}}(x) \, dx + \epsilon/2 \\
   &\leq \epsilon,
\end{align*}
where the last inequality holds for $b_2$ sufficiently close to $b_1$. 
\end{proof}

\begin{proof}[Proof of Theorem~\ref{theo:boosting_type_I_II}]
Ville's inequality implies that 
\begin{align}
    \mathbb{P}_0(\exists t\in \mathbb{N}: \Lambda_t^{\boost}\geq 1/\alpha) \leq \alpha \label{eq:two-sided_typeI}\\
    \mathbb{P}_1(\exists t\in \mathbb{N}: \Lambda_t^{\boost, \mathrm{inv}}\geq 1/\beta) \leq \beta.  \label{eq:two-sided_typeII}
\end{align}
\eqref{eq:two-sided_typeI} immediately implies the type I error control of $\delta^{\boost}$. To deduce type II error control of $\delta^{\boost}$ from \eqref{eq:two-sided_typeII}, it remains to show that $\Lambda_t^{\boost} \leq \nu_t$, $t\in \mathbb{N}$, implies $\Lambda_t'^{\boost, \textrm{inv}}\geq 1/\beta$ for some $t'\in \mathbb{N}$. To see this, just note that $\Lambda_t^{\boost} \leq \nu_t$ implies that there is a $t'\leq t$ such that
\begin{align*}
& \prod_{i=1}^{t'} b_i \lambda_i(X_i) \leq \beta \prod_{i=1}^t b_i \prod_{i=1}^t b_i^{\textrm{inv}} \ \land \ \prod_{i=1}^{t''} b_i \lambda_i(X_i) <1/\alpha  \forall t''< t' \\
\Leftrightarrow &   \prod_{i=1}^{t'} \frac{b_i^{\textrm{inv}}}{\lambda_i(X_i)} \geq \frac{1}{\beta} \ \land \ \alpha \prod_{i=1}^{t''} b_i \prod_{i=1}^{t''} b_i^{\mathrm{inv}} < \prod_{i=1}^{t''} \frac{b_i^{\mathrm{inv}}}{\lambda_i(X_i)}  \forall t''< t' \\
 \implies & \Lambda_{t'}^{\boost, \mathrm{inv}} \geq 1/\beta.
\end{align*}

\end{proof}

\begin{proof}[Proof of Proposition~\ref{prop:boosting_equal_2side}]
    Since the expected values in \eqref{eq:boost_original} and \eqref{eq:boost_inverse} are continuous functions of the boosting factors if $\lambda_t(X_t)$ is continuously distributed (see proof of Proposition~\ref{prop:solvability_2side} for more insights), the claim follows immediately by the Poincar\'e-Miranda theorem \citep{miranda1940osservazione}. 
\end{proof}

\begin{proof}[Proof of Proposition~\ref{prop:improve_cons_SPRT}]
    Since $b_t\geq 1$ for all $t\in \mathbb{N}$, we immediately have that $\Lambda_{t}\geq 1/\alpha$ implies $\Lambda_{t}^{\boost} \geq 1/\alpha$. Now suppose $\Lambda_{t}\leq \beta $ for some $t$. Let's first consider the case $\nu_t=\beta \prod_{i=1}^t b_i \prod_{i=1}^t b_i^{\mathrm{inv}}$. Since $b_i^{\mathrm{inv}}\geq 1$ for all $i\leq t$, $\Lambda_{t}\leq \beta $ immediately implies $\Lambda_{t}^{\boost}\leq \nu_t $. Now consider the case $\nu_t=1/\alpha$. Then $\tau_{\boost}\leq t$, since we either have $\Lambda_t^{\boost} =0$ or $\Lambda_t^{\boost} =1/\alpha$.
\end{proof}

\begin{proof}[Proof of Proposition~\ref{prop:type_I_II_two_sided}]
In the same manner as in Proposition~\ref{prop:one_sided_null}, one can show that the boosting and inverse boosting factors $b_t$ and $b_t^{\mathrm{inv}}$ derived under $\theta_0$ and $\theta_1$ are also valid for all other $\theta\leq \theta_0$ and $\theta\geq \theta_1$, respectively. With this, the claim follows immediately by Theorem~\ref{theo:boosting_type_I_II}.
\end{proof}

\end{document}